\documentclass[11pt]{article}
\usepackage{amssymb,amsmath,amsthm,url}
\usepackage{graphicx}
\usepackage{labels}
\usepackage{algorithmic}
\usepackage{algorithm}
\usepackage{verbatim}
\usepackage{cancel}
\usepackage{tabularx}

\usepackage{tikz}
\usetikzlibrary{arrows}
\usepackage{tkz-graph}
\usetikzlibrary{arrows,%
                shapes,positioning}

\usepackage{pstricks}
\usepackage[tiling]{pst-fill}

\usepackage{mflogo}
 \input{random.tex} 
 \usepackage{url} 
 \usepackage{multido} 
\usepackage{hyperref}
\usepackage{url}

\topmargin by -\headsep

\newtheorem{theorem}{Theorem}

\newtheorem{lemma}[theorem]{Lemma}

\theoremstyle{definition}

\newtheorem{definition}{Definition}
\newtheorem{remark}{Remark}

\newtheorem{example}{Example}

\title{\bf Impossibility Theorems and the Universal Algebraic Toolkit}

\author{
        \textsc{Mario Szegedy} and
        \textsc{Yixin Xu}
        \mbox{}\\ %
        Department of Computer Science\\
        Rutgers, The State University of New Jersey\\
        \mbox{}\\ %
        \normalsize
            \texttt{szegedy/yixinxu}
        \normalsize
            \texttt{@cs.rutgers.edu}
}

\date{\today }

\begin{document}

\maketitle

\begin{abstract}
We elucidate a close connection between the Theory of Judgment Aggregation
(more generally, Evaluation Aggregation),
and a relatively young but rapidly growing field of universal algebra, that was primarily developed 
to investigate constraint satisfaction problems. 
Our connection yields a full classification
of non-binary evaluations into possibility and impossibility domains
both under the idempotent and the supportive conditions.
Prior to the current result E. Dokow and R. Holzman nearly classified
non-binary evaluations in the supportive case, by combinatorial means.
The algebraic approach gives us new insights to the easier binary case as well, which had been
fully classified by the above authors.
Our algebraic view lets us put forth a suggestion about a strengthening of the Non-dictatorship criterion,
that helps us avoid ``outliers'' like the affine subspace. 
Finally, we give upper bounds on the complexity 
of computing if a domain is impossible or not (to our best knowledge no finite time bounds were given earlier).
\end{abstract}


\newpage


\section{Introduction}

The goal of Judgment Aggregation is to investigate the existence 
or nonexistence of functions that ``aggregate individual
sets of judgments on multiple, logically connected propositions into collective sets of judgments'' \cite{LP02}. 
Arrow's impossibility theorem, which can be interpreted in the judgment aggregation framework, is a powerful example 
showing that even in simple cases 
we cannot meaningfully aggregate individual opinions
into a social opinion.

\medskip

Elad Dokow and Ron Holzman \cite{DH10a,DH10b} have studied an 
elegant generalization of Judgment Aggregation, that we call 
after their title {\em Aggregation of Evaluations}. 
Let $J$ be a finite set of issues and $D$ be a finite set of possible positions/opinions (like `yes', `no' etc.).
Without loss of generality we assume that
\[
J = [m]=\{1,\ldots,m\}.
\]
An evaluation $(v_{1},\ldots,v_{m})\in D^{m}$ assigns a position in $D$ to each $j\in [m]$.
The binary case, when $D=\{0,1\}$
has received special attention \cite{Wilson75, Rubinstein86, DH10a}.
Our fundamental object is the {\em domain} $X\subseteq D^{m}$ 
of {\em feasible evaluations}, these are the evaluations (i.e. opinion-combinations) that we 
allow for the voters to choose from. 

\begin{example}
Assume that during a murder trial  the members of the jury 
have to vote on two issues: 1. the suspect had a knife; 2. the suspect was the murderer; with 
taking a position either `yes' or `no' on each of the issues. Each member must take a position on both issues, 
but the jury agrees that the position-combination (1: no, 2: yes) should not be valid: neither as an individual vote
nor when we {\em aggregate} the votes.  Thus $X = \{\mbox{(no,no), (yes,no), (yes,yes)}\}$.
\end{example}

\medskip

{\bf Aggregators.} When $n$ members  of a society take a position on all of the $m$ issues, 
and each member's vote is from $X$, we get a {\em profile} vector 
\[
(x^{(1)},\ldots,x^{(n)})\in X^{n}.
\]
Our goal is to design a function $f:X^{n}\rightarrow X$ that takes
profile vectors into single elements of $X$. Such functions are called
{\em aggregators}.  The aggregators we shall consider must satisfy three conditions that
come directly from Arrow's famous conditions, he has identified
while studying the special case of preference list
aggregation. Before we describe them we remark that
\[
x^{(i)}=(x^{(i)}_{1},\ldots, x^{(i)}_{m})\in D^{m}\;\;\;\;\;\mbox{for}\; i = 1,\cdots,n,
\] 
are vectors themselves: $x^{(i)}_{j}$ is the $i^{\rm th}$ voter position on the $j^{\rm th}$ issue.
Thus the profile is a vector of vectors.
The output of $f$ is a vector in $D^{m}$,
representing the aggregated positions on the $m$ issues. 
The latter vector must also belong to $X$. 

\medskip

The first and key condition
is that each issue ought to be aggregated independently from the others 
(also called point-wise aggregation or Issue by Issue Aggregation):

\begin{description}
\item[Issue-by-Issue Aggregation (IIA):] There are functions $f_{j}: D^{n}\rightarrow D$ ($1\le j\le m$) such that 
for every $(x^{(1)},\ldots,x^{(n)})\in X^{n}$:
\[
f(x^{(1)},\ldots,x^{(n)}) = \left(f_{1}\left(x^{(1)}_{1},\ldots,x^{(n)}_{1}\right),\ldots,f_{m}\left(x^{(1)}_{m},\ldots,x^{(n)}_{m}\right)\right)
\]
\end{description}

There is a nice way to visualize the IIA property via the picture

\medskip

\begin{center}
\begin{tabular}{ccccc}
 $x^{(1)}_{1}$ & $\cdots$ & $x^{(1)}_{m}$  & $\in$ & $X$ \\
                      & $\vdots$ &                        &         &         \\
\rule[-2ex]{0pt}{0pt}  $x^{(n)}_{1}$ & $\cdots$ & $x^{(n)}_{m}$   & $\in$ & $X$ \\\hline
\rule{0pt}{3ex}   $\downarrow_{f_{1}}$       & $\cdots$ & $\downarrow_{f_{m}}$       &  \\
\rule{0pt}{3ex}   $x_{1}$       & $\cdots$ & $x_{m}$       & $\in$ & $X$ \\
\end{tabular}
\end{center}

\medskip

Above 
we aggregate column $j$ (where $j\in [m]$ is an issue) by function $f_{j}$. The condition that $f$ takes $X^{n}$ to $X$ is
equivalent to saying that if
each row belongs to $X$, then so does the aggregated row.
The component aggregate functions should work in unison to accomplish this.
We have adopted the term ``Issue by Issue Aggregation'' coined by E. Dokow and R. Holzman \cite{DH10a,DH10b},
which is in this generalized context more fitting than the commonly used 
``Independence of Irrelevant Alternatives'' expression,
with the benefit that the acronym remains the same.

\medskip

{\bf Uniqueness of the IIA decomposition.} The representation of $f:X^{n}\rightarrow X$ as $(f_{1},\ldots,f_{m})$ is clearly not unique for instance when $D$ 
contains any element that does not occur as a constituent in any $x\in X$. In order to avoid non-uniqueness of the $f_{j}$s we define
\[
D_{j} = {\rm pr}_{j} X = \{ u_{j} \mid (u_{1},\ldots, u_{m})\in X\}. 
\]
If we define $f_{j}$ on $D_{j}^{n}$ instead of $D^{n}$, it is easy to see that $f_{j}$ becomes unique. 
Throughout the paper we shall assume this.

\medskip

Next we describe the two other conditions (besides IIA) Arrow has imposed on an aggregator $f:X^{n}\rightarrow X$.

\begin{description}
\item[Idempotency (or Unanimity):] $f(x,\ldots,x)=x$ for every $x\in X$.
\end{description}

\begin{lemma}
An IIA aggregator $f=(f_{1},\ldots,f_{m})$ is idempotent if and only if every $f_{j}$ is idempotent
in the universal algebraic sense, i.e.
\[
\forall \; 1\le j\le m \;\;\;\; \;\; \forall \; u\in D_{j} : \;\;\;\;\;\;\; f_{j}(u,\ldots,u) = u
\]
\end{lemma}

\begin{description}
\item[Non-dictatorship:] Aggregator $f:X^{n}\rightarrow X$ is a dictatorship if there is a $1\le k\le n$ such that for every $(x^{(1)},\ldots,x^{(n)})\in X^n$ we have 
$f(x^{(1)},\ldots,x^{(n)}) = x^{(k)}$. Otherwise the Non-dictatorship condition holds for $f$. \vspace{0.1in}
\end{description}

\begin{lemma}
An IIA aggregator $f=(f_{1},\ldots,f_{m})$ is a dictatorship if and only if there is a $1\le k \le m$ such that each $f_{j}$ is a {\em projection} on 
the $k^{\rm th}$ coordinate in the universal 
algebraic sense:
\[
\forall \; 1\le j\le m \;\;\;\; \;\; \forall \; u_{1},\ldots,u_{n} \in D_{j}  : \;\;\;\;\;\;\; f_{j}(u_{1},\ldots,u_{n}) = u_{k}
\]
\end{lemma}

\begin{definition}[{\bf Impossibility/Possibility domains}]
We call an $X\subseteq D^{m}$ a possibility domain (after Arrow) 
with respect to the IIA + Idempotency + Non-dictatorship
conditions if for some $n\ge 2$ an aggregator function $f$ for $X$ with arity $n$ exists that satisfies the three said conditions. 
Otherwise $X$ is an impossibility domain.
\end{definition}

In this article we completely characterize impossibility domains with respect to the 
IIA + Idempotency + Non-dictatorship and also for the case when 
Idempotency is replaced with 

\begin{description}
\item[Supportiveness:] $f:X^{n}\rightarrow X$ is supportive if for every $x^{(1)},\ldots,x^{(n)}\in X^{n}$ and
every $1\le j\le m$ we have that 
$f(x^{(1)},\ldots,x^{(n)})_{j} \in  \{x^{(1)}_{j},\ldots,x^{(n)}_{j}\}$.
\end{description}

\begin{lemma}
An IIA aggregator $f=(f_{1},\ldots,f_{m})$ is a supportive if and only if every $f_{j}$ is {\em conservative} in the universal 
algebraic sense:
\[
\forall \; 1\le j\le m \;\;\;\; \;\; \forall \; u_{1},\ldots,u_{n} \in D_{j}  : \;\;\;\;\;\;\; f_{j}(u_{1},\ldots,u_{n}) \in \{u_{1},\ldots, u_{n}\}
\]
\end{lemma}

Supportiveness implies Idempotency, but not vice versa.

\medskip

Prior to our result a full characterization of all impossible binary domains was obtained
under IIA + Idempotency + Non-dictatorship  in \cite{DH10a}.
They have extended their work to the non-binary case, but have obtained 
only a partial characterization, and only under the IIA + 
Supportiveness + Non-dictatorship conditions.

\medskip

In our characterization we exploit a Galois Connection
discovered by D. Geiger \cite{Geiger68}.
Geiger's duality theorem allows us to characterize the nonexistence of aggregators 
with existence of gadgets --- existential logical expressions created from $X$
and a few very basic relations, like assignment giving or unary relations. 
(Gadgets are also sometimes called conjunctive queries in the literature.)
Our main results are stated in 
Theorem \ref{nonbinaryclassthm} and Theorem \ref{gth}.
Our main contribution is not any new technology but rather pointing to a so far unexplored very broad connection.
Simplicity only works to the new connection's favor.

\medskip

Based on the algebraic view, we are also in the position to strengthen the Non-dictatorship condition 
to exclude the possibility of what experts agree are outliers such as linear subspaces.

\medskip

The theory we use here is getting increasingly familiar to computer scientists
because it gives a powerful machinery to tackle the 
well-known dichotomy conjecture of Feder and Vardy \cite{FV98} for 
Constraint Satisfaction Problems (CSP).
The hope is that the connection will allow researchers to exploit the vast material that CSP research has 
created, in proving impossibility theorems.

\section{Background}\label{background_sec}
The elegant
general combinatorial framework described in the introduction
which serves as the basis of the current paper
was laid down by E. Dokow and and R. Holzman in \cite{DH10a,DH10b}. 
We have named it ``Aggregation of Evaluations" after their titles.
In their two breakthrough results they make decisive advances towards classifying 
impossibility domains, i.e. those $X$ from which we cannot 
aggregate opinions. In particular, they completely settle the binary case. To explain 
their results we need some definitions. 

\begin{definition}
We call $X$ non-degenerate if $|{\rm pr}_{j} X| > 1$ for every $1\le j\le m$. Since the issues 
where degeneration occurs can be trivially aggregated, without loss of 
generality we can assume that $X$ is non-degenerate.
\end{definition}

\begin{definition}[Blockedness graph and MIPE]\label{blockgraph} 
The blockedness graph for domain $X\subseteq\{0,1\}^{m}$ 
is the following directed graph on the vertex set $V = [m] \times \{0,1\}$:
There is a directed edge from $(k,\sigma)\in V$ to $(\ell,\rho)\in V$ where $k \neq \ell$ if and only if there are:
(i.) a subset $S\subseteq [m]$ such that $k,\ell \in S$ and (ii.) a (partial-)evaluation $u: S\rightarrow \{0,1\}$ 
with $u_{k} = \sigma$ and $u_{\ell} = \neg \rho$ such that 
there is no extension of $u$ to any full evaluation $x$ in $X$, but if we flip any bit 
of $u$ then the resulting partial evaluation extends to some element of $X$.
The above partial assignment $u$ is called a MIPE (minimally infeasible partial evaluation). 
\end{definition}

\begin{definition}[Total blockedness]\label{blockdef} 
Domain $X\subseteq\{0,1\}^{m}$ has the total blockedness condition
if and only if the blockedness graph is strongly connected.
\end{definition}

A result leading to \cite{DH10a} was that of
Nehring and Puppe \cite{NP02}. They have obtained a complete classification of 
binary impossibility domains, when a monotonicity condition is added to the usual conditions.
An aggregator is said to be monotone if for every issue $j$
in any situation the aggregate position on issue $j$ does not change if a 
voter decides to switch his/her position on the $j^{\rm th}$ issue to the current aggregate position.


\begin{theorem}[Nehring and Puppe \cite{NP02}]
A non-degenerate $X\subseteq\{0,1\}^{m}$ is an impossibility domain with respect to IIA + Idempotency + Monotonicity + Non-dictatorship
if and only if $X$ is totally blocked. 
\end{theorem}

The complete characterization of binary evaluations without the monotonicity condition 
was finally given by E. Dokow and R. Holzman:

\begin{theorem}[E. Dokow, R. Holzman \cite{DH10a}] \label{binarycase}
Let $X\subseteq \{0,1\}^m$, non degenerate. Then $X$ is an impossibility domain with respect to 
IIA + Idempotency + Non-dictatorship
if and only if $X$ is totally blocked and is not an affine subspace.
\end{theorem}

Dokow and Holzman has also made significant progress for general $D$ \cite{DH10b}.
In their near-characterization they use a
generalization of total blockedness for non-binary predicates, which is similar but more intricate than the binary
notion, and we postpone the definition to Section \ref{totblsec} (Definition \ref{totbl}).

\medskip

We also need to define a condition on a relation $X$, which is
called  2-decomposability in the universal algebra literature.
We have adopted this term rather than the ``not multiply constrained''
expression for the same concept in \cite{DH10b}.

\begin{definition} \label{defmultiplyconstrained}
For an $m$-ary relation $X$ on $D$ and for $1\le k < \ell \le m$, let
${\rm pr}_{k,\ell} X = \{(u_{k},u_{\ell})\mid (u_{1},\ldots,u_{m})\in X\}$.
A relation $X$ is called 2-decomposable if,
for any tuple $x = (x_{1},\ldots,x_{m})\in D^{m}$ we have $x\in X$ if and only if $(x_{k},x_{\ell})\in {\rm pr}_{k,\ell} X$
for all $1\le k < \ell \le m$. 
\end{definition}

\begin{remark}
Coincidentally, in \cite{DH10b} the `2-decomposable' expression also occurs, but 
with a very different meaning. 
\end{remark}

\begin{theorem}[E. Dokow, R. Holzman \cite{DH10b}]\label{generalcase}
Let $X\subseteq D^{m}$, non-degenerate and non-binary (there is a $1\le j\le m$ 
such that $|{\rm pr}_{j} X|>2$). If $X$ is totally blocked and not 2-decomposable
(see Definition \ref{defmultiplyconstrained}) then $X$ is an impossibility domain
with respect to IIA + Supportiveness + Non-dictatorship.
\end{theorem}

\begin{theorem}[E. Dokow, R. Holzman \cite{DH10b}] \label{nottotallythm}
Let $X\subseteq D^{m}$. If $X$ is non-degenerate and not totally 
blocked then $X$ is a possibility domain
with respect to IIA + Supportiveness + Non-dictatorship.
\end{theorem}

\section{Our Results}\label{our_results}

In spite of the impressive advances due to Dokow and Holzman, 
important questions have remained open:
1. Complete the characterization of the Supportive case, when $|D|>3$.
(In \cite{DH10b} the case $|D|=3$ is resolved.) 2. Settle the $|D|>2$ case with the Idempotency condition.

We take the inspiration from the algebraic theory and combine it with 
ideas from \cite{DH10b} to get a full characterization of the non-binary 
case, with the Supportiveness and the Idempotency conditions.

\begin{theorem} \label{nonbinaryclassthm}
Let $X\subseteq D^{m}$, non-degenerate and non-binary. 
If $X$ is totally blocked then $X$ is an impossibility 
domain with respect to IIA + Supportiveness (Idempotent) + Non-dictatorship
if and only if there is no Supportive (Idempotent) non-dictatorial IIA aggregator 
with at most three ($|D|$) voters.
%
\end{theorem}

The characterization (to our knowledge for the first time)
allows for an algorithmic determination if $X$ is an impossibility domain with respect to the Supportive (Idempotent) conditions.

\medskip

{\bf Gadgets.} 
It turns out that we can characterize impossibility domains in a dual way,
in terms of a set of {\em gadgets}. 
Our new characterization, that provides {\em witnesses} to impossibility, has not been known earlier in the voting theory context. Gadgets (or 
{\em conjunctive queries}) are expressions 
used in reductions between constraint 
satisfaction problems when we translate instances locally (term by term). They are existentially quantified conjuncts of clauses, where each clause 
is a relation. 
The relations in this interpretation are viewed as Boolean-valued functions on not necessarily Boolean {\em variables}. The syntax of a gadget is:
\[
R(\vec{x})= \exists\vec{y}:\; S_{1}(\vec{x},\vec{y})\wedge \ldots \wedge S_{k}(\vec{x},\vec{y})\;\;\;\;\;\mbox{(each $S_{i}$ in effect depends only on subsets of $\vec{x},\vec{y}$)}
\]
Their purpose is to express new relations from a given set of relations.

\medskip

A theorem of D. Geiger \cite{Geiger68}
establishes a connection between the nonexistence of aggregators 
for a set $\Gamma$ of relations
and the existence of $\Gamma$-gadgets (i.e. in which all relations
are from $\Gamma$ or the '=' relation). 
This theorem serves as the backbone of the algebraic theory of constraint
satisfaction problems
developed by P. Jeavons, A. A. Bulatov, A. A. Krokhin, D. A. Cohen, M. Cooper, 
M. Gyssens  \cite{JeavonsCG97,Jeavons98,JeavonsCC98,BulatovJK05} and several other researchers.

\medskip

\begin{figure}[htb] 
\begin{center}
\includegraphics[scale=0.5]{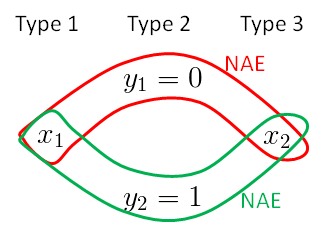}
\caption{The multi-sorted Not All Equal (NAE) gadget to express the 
multi-sorted inequality relation between $x_{1}$ (type 1) and $x_{2}$ (type 3). 
Variables $y_{1}$ and $y_{2}$ are existentially bound (and they are also
set to constant).
This gadget is central in our proof of Arrow's theorem. \label{arrowinequalitygadget}}
\end{center}
\end{figure}

This is our starting point too. 
In universal algebra (IIA) aggregators are called {\em polymorphisms}. 
They differ from the aggregators in our introduction 
in that they are single-sorted:
$f_{1}=f_{2}=\ldots= f_{m}$. {\em Multi-sorted} polymorphisms, i.e.
when the $f_{j}$s can be different, have also been studied in the
algebraic literature \cite{BulatovJ03, Bulatov11}.
When dealing with multi-sorted polymorphisms, all relations and  
gadgets must be multi-sorted as well.
In the multi-sorted world first we must declare a type
for every variable. 
The typing of the variables serves 
the same purpose as in programming languages: when a function is called, the types
of the called variables must match
with the types in the function declaration. 
In the same spirit, every multi-sorted $l$-ary relation $R$ we construct (or given to us) must 
come with a sequence of $l$ (not necessarily different) types. We may call this type declaration.
The typing must be consistent: the the types of the variables involved in any occurrence of a relation in a gadget must
match with the type declaration.

\medskip

In this article we prove a multi-sorted version of Geiger's theorem (see Section \ref{algsec}). This allows us to prove impossibility results 
simply by producing sets of gadgets.
If we want to prove that $X\subseteq \prod_{j}^{m} D_{j}$ is an impossibility domain, we need to construct $X^{+}$-gadgets 
for a certain ``complete'' set of (multi-sorted) relations. Here the `$+$' in the upper index refers to 
the permission to use assignment-giving relations (i.e. `$x=a$,' where $a\in D_{j}$, for type $j$ variables) in our gadgets in addition to $X$.
In all cases $X$ must be viewed as a multi-sorted 
relation, all components (arguments) having different types.
The Idempotency constraint is encoded in the `$+$'  of $X^{+}$.
Alternatively, the Supportiveness condition in the gadget-reformulation translates to the permission of
using arbitrary unary relations in the gadgets. When we allow the latter, the gadget is an
 $X^{\mho}$-gadget.

\begin{theorem}\label{gadgettheorem}
For every $D$ and $m$ there is a fixed finite set ${\cal P} = {\cal P}(D,m)$ 
of multi-sorted relations such that $X\subseteq D^{m}$ is an impossibility domain with respect 
to IIA + Idempotency (Supportiveness) + Non-dictatorship if and only if
we can express every member of ${\cal P}$ 
with a gadget whose conjunct has only (appropriately multi-sorted) $X^{+}$- ( $X^{\mho}$)-clauses. 
In addition, the number of auxiliary variables is upper bounded by some explicit function $\phi(m,|D|)$ (single exponential in $m$). 
\end{theorem}

The theorem in more details is restated in Theorem \ref{gth} and proved in section \ref{specialsec}.

\medskip

Our characterizations
allow us to compute in finite 
time if $X$ is an impossibility domain (both in the Idempotent and in Supportive cases) in two different ways:
either by checking aggregators up to a certain
number of arguments (Theorem \ref{nonbinaryclassthm}), or by checking gadgets up to a certain size (Theorem \ref{gadgettheorem}).

\medskip

One of the main applications of our gadget characterization is that we can prove the impossibility of a domain $X$ by merely
presenting a few gadgets (rather than trying to exclude a large set of aggregators). 
In some cases we can tailor the gadgets to the specific problem, exploiting symmetries. A combined approach where we exclude most aggregators by gadgets
while the rest we treat directly is also possible.
Using gadgets we can show that 
the Pairwise Distinctness relation defined 
by 
\[
\{(u_{1},\ldots,u_{m})\in D^{m}\mid\, u_{k}\neq u_{\ell}\,
(1\le k<\ell \le m)\}
\] 
is an impossibility domain when $|D|>m\ge 2$ and when $|D| = m\ge 3$ 
with respect to the IIA + Idempotency + Non-dictatorship conditions.
In \cite{DH10b} this is proven only under the IIA + Supportiveness + Non-dictatorship conditions, and
\cite{FalikF11} proves the above only when $|D|=m > 2$.

\medskip

The possibility notion we have discussed is considered too generous by some authors.
Several further restrictions were studied e.g. \cite{Kalai02}.
Dokow and Holzman, for instance, question if linear subspaces of $\{0,1\}^m$ that emerge in Theorem \ref{binarycase}
should really be considered possibility domains \cite{DH10a}.
In Section \ref{degreedemocracy} we propose a new aggregator class
that strengthens the notion of Non-dictatorship. 
Our new definition directly comes from the algebraic theory and has many desirable properties.


\medskip

Although the case of binary evaluations was completely settled by Dokow and Holzman,
universal algebra gives a tad more refined form of their theorem:

\begin{theorem} \label{binaryclassification}
Let $X\subseteq \{0,1\}^m$ non-degenerate. Then $X$ is an impossibility domain with respect to 
IIA + Idempotency + Non-dictatorship
if and only if $X$ is totally blocked and it is not an affine subspace.
If $X$ is not totally blocked then for all $1\le j\le m$ one of the following holds:
\begin{enumerate}
\item there is an $f$ such that $f_{j}$ is the semi-lattice operation $u\vee v$ or $u\wedge v$, 
\item there is an $f$ such that $f_{j}$ is the majority operation $(u\vee v)\wedge (v\vee w)\wedge (w\vee u)$,
\item there is an $f$ such that $f_{j}$ is the Mal'tsev operation $u-v+w \mod 2$,
\item $f_{j}$ is a dictatorship for every $f$.
\end{enumerate}
\end{theorem}

\section{The algebraic theory}\label{algsec}

A key observation is that the algebraic theory can be made work in the assignment aggregation context if we turn to 
{\em multi-sorted} relations, gadgets, polymorphisms, etc. instead of the more usual single-sorted ones. 

\medskip

{\bf Multi-sorted relations} differ from usual relations in that each component of the relation we consider has a type. Each type $b$
has a designated range set $D_{b}$. Without loss of generality we may assume that the set of permissible 
types is $[t]=\{1,\ldots,t\}$ where $t$ is a fixed positive integer. The corresponding ranges are
$D_{1},\ldots,D_{t}$.
Let $X\subseteq \prod_{j=1}^{m} D_{\tau_{j}}$ where $\tau_{j}\in[t]$ for $1\le j\le m$.
We denote $X$ with $(X,\tau)$, where $\tau= (\tau_{1},\ldots,\tau_{m})$,
to indicate both that it is multi-sorted and the types of its components.
(In our assignment aggregation setting $t=m$, $D_{j}= {\rm pr}_{j} X$ and the typing of $X$ will be 
$(X,(1,\ldots,m))$.)

\begin{definition}[gadgets, multi-sorted] We fix a type set $[t]$, $\{D_{b}\}_{b\in [t]}$.
A multi-sorted relation $(R,\tau)$
multi-sorted gadget-reduces to a set
$\Gamma$ of multi-sorted relations if there is a multi-sorted gadget expression
\[
R(x_{1},\ldots,x_{k})  = \exists\,  y_{1},\ldots,y_{k'}:  \; R_{1}(z_{1,1},\ldots,z_{1,k_{1}})\wedge\ldots\wedge
R_{p}(z_{p,1},\ldots,z_{p,k_{p}})
\]
where each $R_{i}$ is either from $\Gamma$ or the multi-sorted equality relation 
($x=y$,\, (b,b)) (so ${\rm type}(x)={\rm type}(y)=b$) for some $b\in [t]$. Variables $z_{1,1},\ldots,z_{p,k_{p}}$ are
from the set $\{x_{1},\ldots,x_{k}\}\cup\{y_{1},\ldots,y_{k'}\}$.
\end{definition}

In the algebraic theory of CSPs voting functions are called polymorphisms.
They are very extensively studied and classified according to their algebraic 
properties. In the more usual single-sorted case we are forced to 
aggregate each issue with the same function. In the
multi-sorted case different types are aggregated independently: we have as many 
aggregator functions as types. (Even when for two types $a\neq b$ we have $D_{a}=D_{b}$,
aggregators $f_{a}$ and $f_{b}$ may differ.)
These are exactly the polymorphisms we need in 
the evaluation aggregation setting.

\begin{definition}[multi-sorted polymorphism]\label{mpoldef}
Fix $t$,  $\{D_{b}\}_{b\in [t]}$. Let $(X,\tau)$ be a multi-sorted relation where
$\tau = (\tau_{1},\ldots,\tau_{m})\in [t]^{m}$. Fix $n\ge 1$. A collection 
$f_{b}:D_{b}^{n}\rightarrow D_{b}$ ($1\le b\le t$) of 
functions is said to be a multi-sorted polymorphism of $(X,\tau)$ if 
the tuple $(f_{\tau_{1}},\ldots,f_{\tau_{m}})$ is an IIA 
aggregator for relation $X$,
i.e. it takes $X^{n}$ into $X$.
More generally, a collection 
$f_{b}:D_{b}^{n}\rightarrow D_{b}$ ($1\le b\le t$) of functions is a multi-sorted polymorphism
with respect to a set
$\Gamma$ of multi-sorted relations (each relation is over the same fixed set of types)
if $\{f_{b}\}_{b\in [t]}$ is a multi-sorted polymorphism for each relation in $\Gamma$. 
\end{definition}

{\bf Note.}
The above generalizes the notion 
of single sorted polymorphism, when $t=1$ and $D_{1}=D$, and also the IIA aggregators of the introduction 
with the choice $t=m$, $D_{j}={\rm Pr}_{j}(X)$, and with typing $X$ as $(X,(1,\ldots,m))$.

\begin{definition}[{\rm MPol}] 
Fix $[t]$ and $\{D_{b}\}_{b\in [t]}$.
The set of {\em all} multi-sorted polymorphisms for a set $\Gamma$ of multi-sorted relations is denoted 
by ${\rm MPol}(\Gamma)$.
\end{definition}

\begin{definition}[{\rm MInv}] 
Fix $[t]$ and $\{D_{b}\}_{b\in [t]}$.
The set of {\em all} multi-sorted relations that are kept by a set ${\cal F}$ of multi-sorted aggregator functions is denoted  by 
${\rm MInv}({\cal F})$.
\end{definition}

\begin{definition}[$\langle \cdot\rangle$] 
Fix $[t]$ and $\{D_{b}\}_{b\in [t]}$.
For a set ${\Gamma}$ of multi-sorted relations we define:
\begin{eqnarray*}
\langle {\Gamma} \rangle  = \{ (R,\tau) \mid \mbox{ $(R,\tau)$ multi-sorted gadget- 
reduces to ${\Gamma}$}\}   
\end{eqnarray*}
\end{definition}

Now we can state our multi-sorted Geiger theorem:

\begin{theorem}\label{Mth} 
Fix $[t]$ and $\{D_{b}\}_{b\in [t]}$, and let ${\Gamma}, {\Gamma}'$ be (possibly infinite) sets of multi-sorted relations
and let ${\cal F}, {\cal F}'$ be (possibly infinite) sets of multi-sorted aggregators. Then
\begin{enumerate}
\item ${\rm MInv}({\rm MPol}({\Gamma})) = \langle {\Gamma} \rangle$
\item ${\Gamma}\subseteq {\Gamma}' \; \Longrightarrow \; {\rm MPol}({\Gamma}')\subseteq {\rm MPol}({\Gamma})$
\item ${\cal F}\subseteq {\cal F}' \; \Longrightarrow \;  {\rm MInv}( {\cal F}')\subseteq {\rm MInv}({\cal F})$
\end{enumerate}
\end{theorem}

\begin{proof}
Let us first recall the single-sorted Geiger theorem:

\begin{theorem}[D. Geiger \cite{Geiger68}]\label{geig}
Fix $D$, and let ${\Gamma}, {\Gamma}'$ be (possibly infinite) sets of relations on $D$
and let ${\cal F}, {\cal F}'$ be (possibly infinite) sets of functions, such that each function is from some $D$ power to $D$ (i.e. aggregator functions for $D$). Then
\begin{enumerate}
\item ${\rm Inv}({\rm Pol}({\Gamma})) = \langle {\Gamma} \rangle$
\item ${\rm Pol}({\rm Inv}({\cal F})) = [{\cal F}]$
\item ${\Gamma}\subseteq {\Gamma}' \; \Longrightarrow \; {\rm Pol}({\Gamma}')\subseteq {\rm Pol}({\Gamma})$
\item ${\cal F}\subseteq {\cal F}' \; \Longrightarrow \;  {\rm Inv}( {\cal F}')\subseteq {\rm Inv}({\cal F})$
\end{enumerate}
\end{theorem}

In order to make the translation of our multi-sorted version to the above single-sorted one we first define
\[
D = D_{1} \dot{\cup} D_{2} \dot{\cup} \cdots \dot{\cup} D_{t}
\]
where
$D_{b}$ is the range for type $b$. 
If originally the $D_{b}$s are not disjoint, we make them disjoint without the loss of generality.
A non-empty multi-sorted relation 
$X\subseteq D_{\tau_{1}}\times \cdots \times D_{\tau_{m}}$  can be now interpreted as the single-sorted relation
$X_{D}\subseteq D^{m}$. We remark that viewing them as sets, $X$ and $X_{D}$ are exactly the same. The index $D$ in $X_{D}$ is only a reminder that we view $X_{D}$ 
as a single sorted relation over domain $D$,
while we view $X$ as $(X,\tau)$.
Since the $D_{b}$s are disjoint, from any such $X_{D}\neq\emptyset$
we can recover the types of every coordinate (the components of a single element of $X_{D}$ already give this information).
If ${\Gamma}$ is a set of multi-sorted relations over a fixed type-set $[t]$, let ${\Gamma}_{D}$ be the set of those $X_{D}$s 
that $X\in \Gamma$.

\medskip

For $b\in[t]$ we introduce the unary relation $T_{b}$ on $D$ (in the single sorted world):
\[
T_{b}(u) \; \longleftrightarrow \; u\in D_{b}
\]
In other words, $T_{b}(u)$ expresses that ``$u$ has type $b$ in the multi-sorted world.''
For the set $\{T_{1},\ldots, T_{t}\}$ of relations we introduce the notation $\Theta$.

\begin{definition}\label{compdef}
Let $\Delta$ be any set of relations on $D$ such that $\Theta\subseteq \Delta$. 
Then for any polymorphism $f: D^n\rightarrow D$ of $\Delta$ and any  $b\in [t]$
we can define $f_{b}: D_{b}^{n}\rightarrow D_{b}$
as $f_{b}(x) = f(x)$ on $D_{b}^{n}$.
\end{definition}

{\bf Note.} We know that $f$ on $D_{b}^{n}$ takes value from $D_{b}$ since it is an aggregator of $T_{b}\in \Delta$.

\medskip

We have now: 

\begin{lemma}\label{mgreq}
Let ${\Gamma}$ be any set of multi-sorted relations over a fixed type-set $[t]$ and with the notations as before. 
Then the following are equivalent:
\begin{enumerate}
\item $(f_{1},\ldots,f_{t})$ is a multi-sorted polymorphism for $\Gamma$;
\item the sequence $f_{1},\ldots,f_{t}$ arises, as in Definition \ref{compdef}, from some polymorphism $f$ of $\Delta = {\Gamma}_{D}\cup \Theta$.
\end{enumerate}
\end{lemma}

We do not prove this easy lemma. Returning to the proof of Theorem \ref{Mth}, the only challenge is to prove
1. since 2. and 3. are obvious. It is also follows from known composition lemmas
(polymorphisms compose, as do gadgets) that 
\begin{eqnarray*}
{\rm MInv}({\rm MPol}({\Gamma})) \supseteq \langle {\Gamma} \rangle \\
\end{eqnarray*}
We have to show the containment in the other direction.
Theorem \ref{geig} gives that  
\[
{\rm Inv}({\rm Pol}(\Gamma_{D}\cup\Theta)) = \langle \Gamma_{D}\cup \Theta \rangle_{D}
\]
Here the subscript in $\langle\cdot\rangle_{D}$ refers to that gadget generation is taken in the single sorted world.
Our proof scheme is to relate $\langle {\Gamma} \rangle$ to $\langle \Gamma_{D}\cup \Theta \rangle_{D}$ and ${\rm MInv}({\rm MPol}({\Gamma}))$  to ${\rm Inv}({\rm Pol}(\Gamma_{D}\cup\Theta))$.

\medskip

Which non-empty relations can we generate from $\Gamma_{D}\cup \Theta$, i.e. what are the elements of $\langle \Gamma_{D}\cup \Theta \rangle$?
A gadget with this set of generators is of the form
\begin{equation}\label{fff}
R(x_{1},\ldots,x_{k})  = \exists\,  y_{1},\ldots,y_{k'}:  \; R_{1}(z_{1,1},\ldots,z_{1,k_{1}})\wedge\ldots\wedge
R_{p}(z_{p,1},\ldots,z_{p,k_{p}})
\end{equation}
where each $R_{i}$ is either from $\Gamma_{D}$ or from $\Theta$ or the $=$ relation. 

\begin{definition}
We say that (an $x$- or $y$-) variable $z$ in (\ref{fff}) has type $b$ if 
whenever the conjunct holds for an assignment (without the existential quantifier), $z$ has a value in $D_{b}$. In different 
words, adding $T_{b}(z)$ to the conjunction would not eliminate any of the satisfying assignments of the conjunct.
\end{definition}
Since a $\Gamma_{D}\cup\Theta$ gadget (recall, this is a single-sorted gadget) has variables only over $D$, there is no a priori type restriction (other than the entire $D$) on any variable. 
Nevertheless, if a variable $z$ is involved in a relation $S_{D}$, where $S$ is a relation from $\Gamma$, then $S$ gives a type $b$ to that variable: $S_{D}$ never holds if $z\not\in D_{b}$, so we might 
as well restrict $z$ to $D_{b}$. Assume now that there is a variable $z'$ such that the right hand side of (\ref{fff}) contains the relation $z=z'$ with an already restricted $z$.
Then $z'$ must also be from $D_{b}$. Of course, a chain of such equations also enforces a type on the variable in the end of the chain. 
Finally, any relation $T_{b}(z)$ enforces a type $b$ on $z$.
In summary, we can assign type $b$ to variable $z$ if 
\begin{enumerate}
\item $T_{b}(z)$ occurs in the gadget;
\item $z$ occurs in a constraint from $\Gamma_{D}$ with type $b$;
\item There is a chain of equality relations that starts form any variable restricted to type $b$ (by 1. or 2.) that leads to $z$.
\end{enumerate}
We note that a variable cannot have two different (i.e. contradicting) types defined this way. Any contradiction in 
types would make $R$ unsatisfiable (i.e. the empty relation). 
Therefore $R$ is equivalent to a direct product of a ``typed part'' and an ``untyped part'':
\[
R = R_{\rm typed} \times \underbrace{(w_{1,1}=\cdots = w_{1,s_{1}}) \times \cdots \times (w_{l,1}=\cdots = w_{l,s_{l}})}_{\rm untyped \; part}
\]
The $w_{i,j}$s are variables and the untyped (or typed) part might not be there. Point 1.-3. also give that $R_{\rm typed}$ is a syntactically recognizable part 
of $R$ that arises by deleting some variables and terms from the right hand side. Although $R_{\rm typed}$ is a $\Gamma_{D}$-gadget (recall that $\Gamma_{D}$ is an {\em untyped} set of relations
constructed from the {\em typed} set of relations, $\Gamma$), because of its syntax we can also read it as a
(typed) $\Gamma$-gadget, and in addition with the exact same semantics (meaning that $R_{\rm typed}$ as a set is the exact same set as when we read the formula as a $\Gamma$-gadget).
So $R_{\rm typed} \;\mbox{[viewed as a multi-sorted relation]}\; \in \langle {\Gamma} \rangle$.

\medskip

Next we claim that
\[
R\in {\rm MInv}({\rm MPol}({\Gamma}))  \;\;\;\;\longrightarrow \;\;\;\; R_{D}\in {\rm Inv}({\rm Pol}(\Gamma_{D}\cup\Theta)).
\]
The left hand side reads that $R$ is kept by all
 multi-sorted polymorphism $(f_{1},...,f_{t})\in {\rm MPol}({\Gamma}))$.
By Lemma  \ref{mgreq} we have that $f\in {\rm Pol}(\Gamma_{D}\cup\Theta)$
if and only if its component-sequence (as in Definition \ref{compdef})
$(f_{1},...,f_{t})\in {\rm MPol}({\Gamma})$.
 So $R_{D}$ is kept by all
 $f\in {\rm Pol}(\Gamma_{D}\cup\Theta)$, proving our claim.
 

\medskip

By the single sorted Geiger theorem (applied to  $\Gamma_{D}\cup\Theta$)  we have then that since $R_{D} \in {\rm Inv}({\rm Pol}(\Gamma_{D}\cup\Theta))$, we also have that
$R_{D}\in \langle {\Gamma}_{D}\cup\Theta \rangle$.
Two paragraphs earlier we have seen, that then $R_{D}$ must be of the form $R_{\rm typed} \times (\alpha_{1,1}=\cdots = \alpha_{1,s_{1}}) \times \cdots \times (\alpha_{l,1}=\cdots = \alpha_{l,s_{l}}) $.
But since all component of $R_{D}$ are typed, only the typed part is there ($R_{D}=R_{\rm typed}=R$ as sets). So $R \in \langle {\Gamma} \rangle$.
\end{proof}

\section{Gadgets characterize Impossibility}\label{building}

The Galois connection of Geiger (and its multi-sorted version) sends the message that the more gadgets we can create from a set of relations 
(or from a single relation), the smaller set of aggregators this set of relations has.
An impossibility domain $X$ does not have any non-trivial aggregator, therefore $X$
is expected to generate all relations.
Of corse, what excites us more is the converse. If we can write all relations as gadgets made from $X$ and some simple relations
then $X$ must be an impossibility domain. There is however some work ahead:
\begin{enumerate}
\item We need to understand the role of the Idempotency (Supportiveness) conditions.
\item We want to find a minimal (for algorithmic reasons, but also for convenience) gadget set that already implies the impossibility of $X$.
\end{enumerate}

\subsection{The Idempotency and Supportiveness conditions}

The Idempotency and Supportiveness conditions correspond to adding extra relations
to our base set of relations from which we build the gadgets that prove the impossibility of $X$.
Our base set is originally $\{X\}$.
For the Idempotency condition we add
all {\em Assignment-giving relations}
and for the Supportiveness condition we add all {\em Unary relations} (notions defined below).

\begin{definition}[Unary relations, $X^{\mho}$]
A multi-sorted unary relation is $(x\in A,(a))$, where $a\in [t]$ is a type and $A\subseteq D_{a}$.
For a multi-sorted relation $X$ we define
$X^{\mho}$ as the set of relations that consists of $X$ and all unary relations for the type set $[t]$, $\{D_{b}\}_{b\in [t]}$
on which $X$ is defined.
\end{definition}

\begin{definition}[Assignment-giving relations, $X^{+}$] Multi-sorted assignment-giving relations are
unary relations of the form $(x=v,(a))$, where $v\in D_{a}$  (i.e. $|A|=1$).  $X^{+}$ is
the set of relations that includes $X$ and all assignment giving relations for the type set $[t]$, $\{D_{b}\}_{b\in [t]}$
on which $X$ is defined.
\end{definition}

\begin{lemma}\label{idsup}
A multi-sorted function $g=(g_{1},\ldots,g_{t})$, 
where $g_{a}:D_{a}^{n}\rightarrow D_{a}$, is idempotent, (meaning each $g_{a}$ is idempotent)
if and only if it aggregates all assignment giving relations for all types.
It is supportive, 
if and only if it aggregates all unary relations for all types.
\end{lemma}

We omit the simple proof of this known statement. 

\subsection{The Translation of the original problem}

Before characterizing it with gadgets, we devote a short section to the exact notion of `impossibility' in the algebraic language.

\begin{lemma}\label{translation}
The following voting theory and algebraic conditions are equivalent:

\medskip

\begin{center}
\begin{tabular}{p{7.5cm} p{1cm} p{7cm}}
{\bf Algebra} & & {\bf Voting Theory} \\
$X\subseteq D^m$ is a relation & & $X\subseteq D^m$ is a Domain \\
 $f = (f_{1},\ldots,f_{m})$, where $f_{j}: D_{j}^{n}\rightarrow D_{j}$. Here $D_{j} = {\rm pr}_{j} X$ by definition (resticting $D$
 to $D_{j}$ is important). & & IIA aggregator function $f:X^n\rightarrow X$ \\
 $\exists \; 1\le k\le n:\;$ all $f_{j}$ are projections on their $k^{\rm th}$ coordinate.  & &  $f$ is a Dictatorship \\
 All $f_{j}$s are idempotent (supportive)  & & $f$ is idempotent (supportive) \\
\end{tabular}
\end{center}

\medskip

The following voting theory and algebraic problems are equivalent:

\medskip

\begin{center}
\begin{tabular}{p{7.5cm} p{1cm} p{7cm}}
{\bf Algebra} & & {\bf Voting Theory} \\
Determine if $X$ has & & Determine if $X$ has \\
idempotent (supportive)  & & idempotent (supportive) \\
non-dictatorial & & non-dictatorial  \\
multi-sorted polymorphism. & & IIA aggregator. \\
\end{tabular}
\end{center}

\medskip

$X$ is an {\em impossibility domain} if the answer to the above question is ``No.''
\end{lemma}

\subsection{Characterization of Impossibility Domains in terms of gadgets}

Applying Lemmas \ref{idsup}, \ref{translation} 
and our multi-sorted Geiger's theorem (Theorem \ref{Mth}) it is not hard to show: 

\begin{lemma}
Domain $X\subseteq D^m$ is impossible with respect to IIA + Idempotency (Supportiveness) + Non-dictatorship
if and only if {\em all} multi-sorted relations can be generated as multi-sorted
$X^{+}$ ($X^{\mho}$) -gadgets. Here $t=m$, $D_{j}= {\rm pr}_{j} X$ for $1\le j\le m$ and the typing of $X$ is 
$(X,(1,\ldots,m))$. 

\end{lemma}

The set of all (multi-sorted) relations is however infinitely large!
Luckily, we can select a finite subset (actually, in many different ways) that 
generate all relations, and it is sufficient to consider only those. 


\begin{definition}[The Non-Binary OR relation]
We define the multi-sorted relation $R^{u,v}_{k,\ell}$ 
which is unsatisfied only when $x=u$ and $y=v$ simultaneously hold ($x$ has type $k$ and
$v$ has type $\ell$). In formula:
\[
R^{u,v}_{k,\ell}(x,y) = (\neg (x=u \wedge y=v), (k,\ell)).
\]
\end{definition}

\begin{definition}[The multi-sorted Not-All-Equal relation] \label{multiNAE}
The multi-sorted NAE relation on types $a, b, c$ is defined 
only when $D_{a}=D_{b}= D_{c} = \{0,1\}$. In this case
\[
{\rm NAE}_{a,b,c}(x,y,z) = (|\{x,y,z\}|>1, (a,b,c)).
\]
\end{definition}

We are now ready to state our gadget-characterization of impossibility domains $X$:

\begin{theorem}\label{gth}
Let $X\subseteq D^m$ non-degenerate. 
Let $t = m$ and $\tau = (1,\ldots,m)$, and the range for type $j$ be $D_{j}={\rm pr}_{j} X$.
Then $X$ is an impossibility domain with respect to 
IIA + Idempotency + Non-dictatorship
if and only if there are $(X,\tau)^{+}$-gadgets expressing $R^{u,v}_{k,\ell}$
for every $1\le k, \ell\le m;\; u\in {\rm pr}_{k} X,\; v\in {\rm pr}_{\ell} X$.
Furthermore, if $|D_{j}|= 2$ for some $1\le j\le m$, we also need to add the multi-sorted {\rm NAE} gadget
on types $(j,j,j)$. The analogous statement, when we replace ``Idempotency" with ``Supportiveness"
requires to replace $(X,\tau)^{+}$ with $(X,\tau)^{\mho}$.

\end{theorem}


\section{Gadget-power}\label{specialsec}

First we give a proof of Theorem \ref{gth}.
We show the non-trivial direction, i.e. that if we can create all gadgets required by Theorem \ref{gth}, the only polymorphisms that remain are dictatorships.

\begin{remark}
Even a single gadget present in $\langle \Gamma \rangle$ can have strong consequences. 
For instance, if we can construct the equality gadget $(x=y,(a,b))$ between two different types $a$ and $b$ (with a common alphabet),
any multi-sorted polymorphism $(f_{1},\ldots, f_{t})\in {\rm MPol}(\Gamma)$ must have the same 
$a$ and $b$ components. We prove this in the end of the section.
\end{remark}

Our way to prove Theorem \ref{gth} is to start to create gadgets and use them as ``subroutines'' to create even more gadgets.
 The following lemma says that for any fixed $1\le k, \ell <t$ from $R^{u,v}_{k,\ell}\in  \langle \Gamma \rangle$
 one can construct  {\em all} multi-sorted binary relation with type $(k,\ell)$ in $\langle\Gamma\rangle$.

\begin{lemma}\label{allrel}
Let $1\le k,\ell\le t$ and $(S,(k,\ell))$ be any multi-sorted relation contained
in ${\rm pr}_{k} X \times {\rm pr}_{\ell} X$. 
Then $(S,(k,\ell)) = \wedge_{(u,v)\not\in S} \; R^{u,v}_{k,\ell}$.
\end{lemma}

We omit the straightforward proof. Note that nothing stops us setting $k=\ell$ in the lemma,
in which case we get all binary relations ($R \subseteq D_{k}^2$).
What does the presence of these relations in $\langle \Gamma \rangle$ say about the (multi-sorted) polymorphisms for $\Gamma$?
They say a lot.
In fact, if $D_{k}$ has size at least three, already the 
not-equal relation, $(x\neq y, (k,k))$, alone excludes all idempotent (and so all supportive) polymorphisms
other than the dictatorships {\em for that type}:

\begin{lemma}[Not-equal gadget lemma]\label{ne}
Assume that $((x\neq y),\,(k,k)) \in \langle {\Gamma} \rangle$. 
Then for every $f=(f_1,\cdots,f_t)\in {\rm MPol}(\Gamma)$
it must hold that $f_{k}$ is a dictatorship on ${\rm pr}_{k} X$.
\end{lemma}

The above lemma together with Lemma \ref{allrel} imply 
\begin{lemma}\label{threeproj}
Assume that $|{\rm pr}_{k} X|\ge 3$ and the conditions of Theorem \ref{gth} hold, so we can create 
all $R^{u,v}_{k,\ell}$ from $(X,\tau)^{+}$ (resp. from $X,\tau)^{\mho}$). Then the
 $k^{\rm th}$ component of any idempotent (supportive) aggregator $f = (f_1,\cdots,f_t)$ of $X$
must be a dictatorship. 
\end{lemma}

What if $|{\rm pr}_{k} X| = 2$?
Then the NAE gadget can be used to take care of the same thing. 

\begin{lemma}[Not-all-equal gadget lemma]\label{nae}
Assume that the multi-sorted not-all-equal relation $(|x,y,z|>1 \; \wedge \; x,y,z\in {\rm pr}_{k} X,\,(k,k,k))$ gadget-reduces 
to a set $\Gamma$ of multi-sorted relations. Then for every 
$f=(f_1,\cdots,f_t)\in {\rm MPol}(\Gamma)$
it must hold that $f_{k}$ is a dictatorship on ${\rm pr}_{k} X$.
\end{lemma}

Among the conditions of Theorem \ref{gth} one explicitly states that in the case when $|{\rm pr}_{k} X| = 2$
the $(|x,y,z|>1 \; \wedge \; x,y,z\in {\rm pr}_{k} X,\,(k,k,k))$ gadget can be constructed from $(X,\tau)^{+}$ (resp. from $X,\tau)^{\mho}$)
to make Lemma \ref{nae} applicable.

\medskip

Putting Lemmas \ref{threeproj} and \ref{nae} together, we conclude that if the gadgets
promised by Theorem \ref{gth} are all present, then all components of every
multi-sorted polymorphism $(f_{1},\ldots,f_{m})$ are dictatorships on the respective 
${\rm pr}_{j} X$s. 

\medskip

To finish the ``if'' part of Theorem \ref{gth} all we need to show is that 
these dictatorships are controlled by the same index $K$ (the dictator). Assume this is not the case,
and let $k$ and $\ell$ be components such that $f_{k}$ is dictated by the $K^{\rm th}$ voter
and $f_{\ell}$ is dictated by voter $L\neq K$. Because $X$ is non-degenerate 
(and so ${\rm pr}_{k} X$ and ${\rm pr}_{\ell} X$ have size at least two)
and because of K\"onig's theorem
(or simply by basic combinatorics)
there ought to be $u,v\in {\rm pr}_{k} X$ and $u',v'\in {\rm pr}_{\ell} X$ such that 
\begin{enumerate}
\item $u\neq v$ and $u'\neq v'$;
\item There is an element $U$ of $X$ that takes $u$ on $k$ and $u'$ on $\ell$;
\item There is an element $V$ of $X$ that takes $v$ on $k$ and $v'$ on $\ell$.
\end{enumerate}
Let us now aggregate a set of votes, all from $\{U,V\}$, but the $K^{\rm th}$
vote is $U$ and the $L^{\rm th}$
vote is $V$. Then $(f_{1},\ldots,f_{m})$ aggregates this input to $Z$ such that
the $k^{\rm th}$ issue aggregates to $u$ and the $\ell^{\rm th}$ issue aggregates to $v'$.
Notice now that $R^{u,v'}_{k,\ell}$ holds for $U$ and $V$, but not for $Z$, which is a contradiction,
since every aggregator must keep all gadgets constructible from $X^{+}$ ($X^{\mho}$), in particular $R^{u,v'}_{k,\ell}$.
This concludes the ``if'' (non-trivial) part of the proof of Theorem \ref{gth}.

\medskip

In the rest of the section we assume that
for two types $a,b\in [t]$ we have
$D_{a}=D_{b}$ (equivalently, a 1-1 correspondence `$=$' between $D_{a}$, $D_{b}$).
We show that the ability to construct the `$=$' gadget between types $a$ and $b$ from a set  $\Gamma$ of (multi-sorted) relations implies that for
every polymorphism $(f_{1},\ldots,f_{t})$ of $\Gamma$ we have $f_{a}=f_{b}$.
This can be useful, because the algebraic theory is developed mainly for single-sorted polymorphisms.

\begin{lemma}[Equality gadget lemma]\label{main}
Let $\Gamma$ be a set of multi-sorted relations with type set $[t]$.
Assume that for types $a$ and $b$ we have $D_{a}=D_{b}$, and that the multi-sorted 
equality relation $(x=y,\,(a,b))$ gadget-reduces to $\Gamma$.
Then for every $f=(f_1,\cdots,f_t)\in {\rm MPol}(\Gamma)$ it must hold that $f_{a}$ is identical to $f_{b}$.
\end{lemma}
\begin{proof}
We need to prove that for every $u = (u^{1},\ldots,u^{n})\in D_{a}^{n}$ $(= D_{b}^{n})$ it holds that $f_{a}(u) = f_{b}(u)$.
Consider an arbitrary $u\in D_{a}^{n}$. $f$ is a multi-sorted polymorphism of $(x=y,\, (a,b))$. This follows from the facts that
$f=(f_1,\cdots,f_t)$ is a polymorphism of $\Gamma$ and that $(x=y,\,(a,b))$ gadget-reduces to $\Gamma$.
Therefore, since each line of the table

\begin{center}
\begin{tabular}{ccc}
type $a$ &   &  type $b$ \\
              &  &     \\
$u_{1}$ & = & $u_{1}$  \\
$u_{2}$ & = & $u_{2}$ \\
  &\vdots & \\
$u_{n}$ & = & $u_{n}$  \\\hline
\rule{0pt}{3ex} $f_{a}(u)$  & = & $f_{b}(u)$
\end{tabular}
\end{center}

above the solid horizontal line satisfies the $(x=y,\, (a,b))$ relation, we can apply polymorphism $f$ for the two columns of the table. Now, as discussed, 
$f$ must keep the relation $(x=y,\, (a,b))$, so $f_{a}(u) = f_{b}(u)$.
\end{proof}

\section{Example: Arrow's Theorem}\label{arrowex}

Let ${\cal A}=\{A_{1},\ldots,A_{k}\}$ denote a set of $k$ items and let $S_{k}$ be the domain that corresponds to the set 
of $k!$ different linear orders on ${\cal A}$ in a way we describe below.
Each voter must vote for some linear order and the votes have to be aggregated into a single linear order, the ``choice of the society.''

\medskip

Instead of thinking of a linear order as is (like $A_{2}<A_{3}<A_{4}<A_{1}$) we rather represent it as
a sequence of ${k\choose 2}$ binary positions corresponding to questions of the form
\[
A_{1}<A_{2}? \; , \; A_{1}<A_{3}? \; , \;  \ldots \; , \; A_{n-1}<A_{n}?
\]
(see also Section \ref{background_sec}). When `$<$' is a linear order on ${\cal A}$, the answers 
to these questions (0 = no; 1 = yes) uniquely (and even redundantly) encodes `$<$' in the form of a valid {\em evaluation}.
The set of all valid evaluations (these are binary vectors of length ${n \choose 2}$)
is exactly the relation $S_k \subseteq \{0,1\}^{n \choose 2}$. Arrow famously shows:

\begin{theorem}[Arrow \cite{Arrow50}] \label{arrowthm} When $k\ge 3$,
there is no aggregator $f:S_{k}^{n}\rightarrow S_{k}$ for any $n\ge 2$ that satisfies IIA $+$ Idempotency $+$ Non-dictatorship.
\end{theorem}

It is easy to see that the impossibility of $S_{3}$ implies the impossibility of $S_{k}$ for $k\ge3$.
Below we are going to prove Theorem \ref{arrowthm} when $k=3$, using our method of gadgets. In the 
$A_{1}<A_{2}? \; A_{2}<A_{3}? \; A_{3}<A_{1}?$ basis (replacing `$A_{1}<A_{3}?$' with `$A_{3}<A_{1}?$' does not change the problem) we have that:
\[
S_{3} =\{001,010,011,100,101,110\} = {\rm NAE} \;\;\;\;\; (\mbox{Not All Equal})
\]
Since Arrow allows to aggregate different coordinates with different aggregators, we view $S_{3}$ as a multi-sorted relation
with type-set $[3]$.
We need to show 
that the only polymorphisms of the multi-sorted relation set $S_{3}^{+}$
are projections (dictatorships). 
In section \ref{our_results} we have created a gadget for the relation $(\neg (x=y), (1,3))$. Let us denote it by $R_{13}$. 
By symmetry, $S_{3}^{+}$-gadgets also exist for $R_{12}=(\neg (x=y), (1,2))$ and $R_{23}=(\neg (x=y), (2,3))$. Then the gadget
\begin{eqnarray*}
\exists x_{2}: \; R_{12}(x_{1},x_{2})\wedge R_{23}(x_{2},x_{3}) 
\end{eqnarray*}
expresses $(x_{1}=x_{3}, (1,3))$ (see Fig. \ref{arrowequalitygadget}). We can similarly express any
$(x_{a}=x_{b}, (a,b))$ for $1\le a< b\le 3$.
Once we have generated these relations, 
by Lemma \ref{main} we conclude that every multi-sorted polymorphism 
of $S_{3}^{+}$ must be a single sorted polymorphism i.e. of the form $(f_{1},f_{1},f_{1})$,
and it is well known that all single-sorted idempotent polymorphisms of NAE are dictatorships.

\begin{figure}[htb] 
\begin{center}
\includegraphics[scale=0.5]{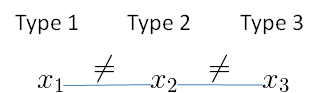}
\caption{The gadget expressing $x_1 = x_3$.}
\label{arrowequalitygadget}
\end{center}
\end{figure}

\section{Example: Pairwise distinctness}\label{distinctnessex}

We continue to illustrate how one can use hand-made gadgets to prove 
impossibility theorems for specific domains:

\begin{theorem}
Let $D$ and $m$ such that $|D|>m =2$ or $|D| = m \ge 3$. Define
\[
X = \{ (x_{1},\ldots,x_{m})\in D^{m} \mid x_{1},\ldots,x_{m} \; \; \mbox{\rm are pairwise distinct} \}
\]
Then $X$ is an impossibility domain with respect to IIA + Idempotency + Non-dictatorship conditions.
\end{theorem}

In \cite{DH10b}, this theorem is proven under the IIA + Supportiveness + Non-dictatorship conditions.
A special case of this theorem, i.e. when $|D| = m \ge 3$, can be derived from results in \cite{FalikF11}. 
We give a gadget proof for it when $m=2$ and $|D|=3$, which is in a sense the hardest setting 
of the parameters. In this case
the relation $X\subseteq [3]^2$ is a binary relation of type $(1,2)$, and we have to show 
that ${\rm MPol}(X^{+})$ contains only projections (dictatorships). 

\medskip

We remark that the problem in this case is essentially equivalent to showing 
that three coloring of bipartite graphs is NP-hard as long as assignment giving constraints are 
also allowed, i.e. we are allowed to declare that ``vertex $v$ has color $c$''.
If we drop the bipartite condition, then the problem is well-known to be NP hard.
The bipartite condition comes from the multi-sorted nature of the problem:
relation $X$ can connect only type 1 variables with type 2 variables.

\medskip

First we create a gadget for the relation
$(\neg (x_{1}=x_{2}=2), (1,1))$. Clearly, by symmetry, then gadgets also exist for 
any $(\neg (x=y=a), (b,b))$, where $1\le a\le 3$ and $1\le b\le 2$.

\begin{figure}[htb] 
\begin{center}
\includegraphics[scale=0.5]{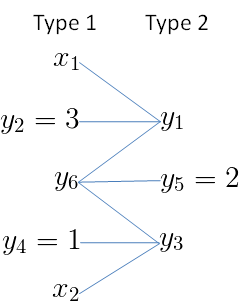}
\caption{The inequality gadget}
\label{inequalitygadget}
\end{center}
\end{figure}

The gadget that achieves this (see Fig. \ref{inequalitygadget}) corresponds to the following formula:
\begin{eqnarray*}
R(x_{1},x_{2}) & = & \exists y_{1},\ldots,y_{6}: \; X(x_{1},y_{1})\wedge X(x_{2},y_{3}) \wedge
X(y_{2},y_{1}) \wedge X(y_{4},y_{3}) \wedge X(y_{6},y_{1})\wedge \\
 & & \wedge X(y_{6},y_{3}) \wedge  X(y_{6},y_{5}) \wedge   (y_{2}=3) \wedge (y_{4}=1)   \wedge (y_{5}=2) 
\end{eqnarray*}

It is not hard to check that $R$ implements the $(\neg (x_{1}=x_{2}=2), (1,1))$ relation.
Let us now denote the relation $(\neg (x=y=a), (b,b))$ by $R_{a}^{b}$.
Then $R_{1}^{b}(x,y) \wedge R_{2}^{b}(x,y)  \wedge R_{3}^{b}(x,y) $ expresses the $(x\neq y, (b,b))$
relation. Finally, the gadget 
\[
T(x_{1},x_{2}) = \exists y_{1},y_{2}: \; (x_{1}\neq y_{1}) \wedge (x_{1}\neq y_{2}) \wedge (y_{1}\neq y_{2}) \wedge X(y_{1},x_{2}) \wedge X(y_{2},x_{2})
\]
expresses $(x_{1}=x_{2}, (1,2))$ (see Fig. \ref{equalitygadget}). Once we have generated this relation, 
by lemma \ref{main} we conclude that every multi-sorted polymorphism 
of $X^{+}$ must be a single sorted polymorphism i.e. of the form $(f_{1},f_{1})$.
It is well known that the only idempotent polymorphisms of the $x\neq y$ (single-sorted) relation are 
the dictatorships.

\begin{figure}[htb] 
\begin{center}
\includegraphics[scale=0.5]{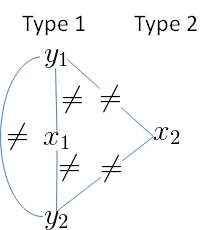}
\caption{The equality gadget}
\label{equalitygadget}
\end{center}
\end{figure}

\section{Binary Evaluations} \label{secbinaryevaluations}

In this section we first translate the total blockades condition to the algebraic language, then using it revisit the case of binary evaluation
giving an alternative proof to the classification Theorem of E. Dokow, R. Holzman, and also indicate the proof of Theorem \ref{binaryclassification}.

\subsection{Total blockedness and its consequences}

\begin{lemma}\label{blockedlemma}
Let $X \subseteq \{0,1\}^{m}$ be totally blocked 
(see Definition \ref{blockdef}) and non-degenerate. Let $(X,\tau)$ be any typing of the variables of $X$ from 
a type-set $[t]$, $D_{a}=\{0,1\}$ for $1\le a \le t$. We also assume that all types are used. Then
for all $1\le a,b \le t$: $(x=y,\, (a,b))$  (multi-sorted) gadget-reduces to $X^{+}$.
\end{lemma}
\begin{proof}
Recall that the total blockedness condition means that on the vertex set $V = [m] \times \{0,1\}$ 
we have a strongly connected graph defined as follows:
There is a directed edge from $(k,\epsilon)\in V$ to $(\ell,\epsilon')\in V$ where $k \neq \ell$ if and only if there are:
(i.) a subset $S\subseteq [m]$ such that $k,\ell \in S$ and (ii.) a (partial-)evaluation $u: S\rightarrow \{0,1\}$ 
with $u_{k} = \epsilon$ and $u_{\ell} = 1 - \epsilon'$ such that 
there is no extension of $u$ to any full evaluation $x$ in $X$, but if we flip any bit 
of $u$ then the resulting partial evaluation extends to some element of $X$.
Let us focus on a directed edge $((k,\epsilon), (\ell,\epsilon'))$ 
as above, with $S=\{k,\ell,s_{1},\ldots,s_{q}\}$ (we fix this $S$ for $k,\ell,\epsilon,\epsilon'$), and create the gadget
\[
E_{k,\ell,\epsilon,\epsilon'}(x_{k},x_{\ell}) = \exists y_{s_{1}},\ldots, y_{s_{q}}, \vec{y} :\;
X(x_{k},x_{\ell},y_{s_{1}},\ldots,y_{s_{q}}, \vec{y}) \wedge (y_{s_{1}}  = u_{s_{1}}) \wedge 
\ldots \wedge (y_{s_{q}}  = u_{s_{q}}) 
\]
Here with some abuse of the notation we tried to indicate, via the indices, the 
variables' positions in $X$. In particular, $\vec{y}$ collects the $m-2-q$ variables of $X$ that 
are not in $S=\{k,\ell,s_{1},\ldots,s_{q}\}$. We remark that the type of a variable is uniquely determined by its position
in $X$. Then we have
\[
E_{k,\ell,\epsilon,\epsilon'}(\epsilon,\epsilon')=1,\;\;\;\; E_{k,\ell,\epsilon,\epsilon'}(\epsilon,1- \epsilon')=0,\;\;\;\; E_{k,\ell,\epsilon,\epsilon'}(1-\epsilon,1 - \epsilon')=1
\]
All three equations follow from the fact that $u: S\rightarrow \{0,1\}$ 
with $u_{k} = \epsilon$ and $u_{\ell} = 1 - \epsilon'$ was a minimally unsatisfying partial assignment,
thus if we change the value of exactly one of the $u_{k},u_{\ell}$, the assignment becomes satisfying.
Consider now a chain 
\[
((k^{0},\epsilon^{0}), (k^{1},\epsilon^{1})),\; ((k^{1},\epsilon^{1}), (k^{2},\epsilon^{2})),\;\ldots, 
((k^{t-1},\epsilon^{t-1}), (k^{t},\epsilon^{t}))
\]
of edges in the blockedness graph,  for which we have generated relations $E_{k^{0},k^{1},\epsilon^{0},\epsilon^{1}},\ldots,E_{k^{t-1},k^{t},\epsilon^{t-1},\epsilon^{t}}$ as above. Create the gadget
\[
R(x_{k^{0}},x_{k^{t}}) = \exists y_{k^{1}},\ldots,y_{k^{t-1}} :\; E_{k^{0},k^{1},\epsilon^{0},\epsilon^{1}}(x_{k^{0}},y_{k^{1}})
\wedge \ldots \wedge E_{k^{t-1},k^{t},\epsilon^{t-1},\epsilon^{t}}(y_{k^{t-1}},x_{k^{t}})
\]
It is easy to see that the typing is consistent.
If we set $x_{k^{0}}=\epsilon^{0}$ then inductively all $y_{k_{i}}$ variables are 
forced to take $\epsilon^{i}$, eventually forcing
 $x_{k^{t}}=\epsilon^{t}$. On the other hand 
 $R(1- \epsilon^{0},1-\epsilon^{t})=1$. To show this it is sufficient to set all $y_{k_{i}}$ variables to $1-\epsilon^{i}$ in the right hand side of the above formula.
 Since the blockedness graph is strongly connected, for any $a,b\in [t]$ with ${\rm type}(x_{k})=a,\; {\rm type}(x_{\ell})=b$ 
 for some $1\le k,\ell\le m$ and for any $\epsilon_{k},\epsilon_{\ell}\in\{0,1\}$ we can build now
 gadget $R_{a,b,\epsilon,\epsilon'}$ that forces $x_{\ell}=\epsilon_{\ell}$ as long as $x_{k}=\epsilon_{k}$ and also 
 permits the $x_{k}=1-\epsilon_{k}$, $x_{\ell}=1-\epsilon_{\ell}$ assignment,
 and a gadget $R'_{a,b,1-\epsilon,1-\epsilon'}$ that forces $x_{\ell}=1-\epsilon_{\ell}$ as long as $x_{k}=1-\epsilon_{k}$ and also permits the 
 $x_{k}=\epsilon_{k}$, $x_{\ell}=\epsilon_{\ell}$
 assignment. Then $R_{a,b,\epsilon,\epsilon'}\wedge R'_{a,b,1-\epsilon,1-\epsilon'}$ implements 
 $x_{k}+\epsilon_{k} = x_{\ell}+\epsilon_{\ell}\;\mod 2$.
 In particular, by choosing $\epsilon{k} = \epsilon{\ell} = 0$ we have implemented the $(x=y,(a,b))$ relation.
\end{proof}

\begin{example}
Let $X=\{000,011,101,110\}= \{xyz\in\{0,1\}^{3}\mid x+y+z=0 \,\mod 2\}$. We also assume that $[t]=3$ and the $i^{\rm th}$
coordinate has type $i$. Then the minimally infeasible partial evaluations (MIPEs) are all those $xyz\in\{0,1\}^3$
for which $x+y+z=1 \; \mod 2$. Then the blockedness graph is the directed complete graph (i.e. directed edges are drawn both ways
for every edge). For $k=1,\; \ell=2,\; \epsilon_{k}=0, \epsilon_{\ell}=1$ we can create the gadget (based on $S=\{1,2,3\}$,  $u=001$):
\[
E_{1,2,0,1}(x_{1},x_{2}) = \exists y:\; X(x_{1},x_{2},y)\wedge (y=1)
\]
This together with $E_{1,2,1,0}$ created from $S=\{1,2,3\}$,  $u=111$
(accidentally, the two gadgets turn out to be the same, so the conjunction remains $E_{1,2,0,1}(x_{1},x_{2})$)
gives the $(x\neq y,(1,2))$ relation, as one can check it directly.
\end{example}

\begin{lemma}\label{eql}
Let $X\subseteq \{0,1\}^{m}$, totally blocked. Then every aggregator $f=(f_{1},\ldots,f_{m})$
which is IIA + Idempotent satisfies that all $f_{j}$'s are identical.
\end{lemma}
\begin{proof}
Let $t=m$ and view $X$ as the multi-sorted relation $(X,(1,2,\ldots,m))$.
Then $(x=y,\, (k,\ell))$ $(1\le k,\ell\le m)$ gadget-reduces to $X^{+}$ by Lemma $\ref{blockedlemma}$. 
Since $f$ is an idempotent IIA aggregator of $X$ we have that $f\in {\rm MPol}(X^{+})$. 
Combining the above two things the statement then follows from
Lemma \ref{main}.
\end{proof}

We remark (although do not use it in the sequel) that total blockedness also generates all non-equal relations:

\begin{lemma}\label{neql}
Let $X\subseteq \{0,1\}^{m}$, totally blocked. and $X$ has type $(X,(1,\ldots,m))$. Then $X$ (as a multi-sorted
relations) generates all relations of the form $(x\neq y,(a,b))$, where $a,b\in[t]$.
\end{lemma}

\subsection{A new proof of Theorem \ref{binarycase}}

We give a new proof of the interesting part of Theorem \ref{binarycase} (E. Dokow, R. Holzman, \cite{DH10a}), i.e.
that when $X$ is totally blocked, it is a possibility domain (with respect to IIA+Idempotency+No Dictatorship) if and only if 
it is an affine subspace. The new proof uses results from the algebraic theory 
of constraint satisfaction problems.

\medskip

Lemma \ref{eql} gives that when $X$ is totally blocked, 
all idempotent multi-sorted polymorphisms of $X$ are single-sorted and we can use
Schaefer's theorem, or more precisely algebraic version of it (Hubie Chen \cite{Chen09}) 
to determine the types of functions in ${\rm Pol}(X)$:

\begin{theorem}[Schaefer, algebraic version] \label{schaeferalg}
Let $D=\{0,1\}$ and $\Gamma$ be a set (single-sorted) relations on $D$.
Then $\Gamma^{+}$ either has one of the following four operations as a polymorphism:
\begin{enumerate}
\item The binary AND operation $\wedge$;
\item The binary OR operation $\vee$;
\item The ternary majority operation ${\rm Maj}_{3}(x,y,z)=(x\wedge y)\vee (x\wedge z)\vee (y\wedge z)$;
\item The Mal'tsev operation $u-v+w \mod 2$. 
\end{enumerate}
Otherwise ${\rm Pol}(\Gamma^{+})$ contains only projections (dictatorships).
\end{theorem}

\medskip

Theorem \ref{schaeferalg} gives that when $X$ is totally blocked and $X$ is not an impossibility domain then one of 
the cases 1.-4. must hold. But \cite{DH10a} proves more, it shows that 
when $X$ is totally blocked, only Case 4. and the default case (i.e. no non-trivial polymorphisms) may occur.
We provide a brief proof of this. We exclude Cases 1.-3. as follows:

\medskip


\medskip

\noindent{\em Excluding 1. and 2:} We show that if $\vee\in {\rm Pol}(X)$, then
in the blockedness graph, no node of the form $(k,1)$ has a directed edge to any 
node of the form $(\ell,0)$, so the blockedness graph cannot be strongly (or anyhow) connected. 
For this it is sufficient to show that every 
MIPE has at most one variable set to 1. Suppose that 
$x_k=x_{\ell}=1$ is part of a MIPE with $k \neq \ell$ and the rest of MIPE evaluates to 
$\alpha$. By the definition we have assignments:

\medskip

\begin{center}
\begin{tabular}{ccccccccc}
  & $x_k$ & $x_{\ell}$ & rest of MIPE & rest&  &  & & \\ \hline\hline
( & $1$ & $1$ & $\alpha$ & any &) & never $\in$ & $X$ & (since it is MIPE) \\\hline
( & $0$ & $1$ & $\alpha$ & some &) & $\in$ & $X$  & (since it was MIPE) \\
( & $1$ & $0$ & $\alpha$ & some &) & $\in$ & $X$  & (since it was MIPE) \\\hline
( & $0\vee 1$ & $1\vee 0$ & $\alpha$ & some & ) & $\in$ & $X$  & (assgnm 2 $\vee$ assgnm 3)  \\
\end{tabular}
\end{center}

\medskip

Then the first and fourth lines of the table contradict to each other. An analogous proof shows that when 
$\wedge\in {\rm Pol}(X)$ the blockedness graph is not strongly connected.

\medskip

\noindent{\em Excluding 3:} Assume that $X$ is totally blocked and ${\rm Maj}_{3}\in {\rm Pol}(X)$. First we show: 
\begin{lemma}\label{atmost2}
If ${\rm Maj}_{3}\in {\rm Pol}(X)$ then every MIPE for $X$ has length at most two.
\end{lemma}
\begin{proof}
Consider a MIPE $S$, which contrary to our assumption has at least 
three elements, $x_{1},x_{2},x_{3}$.
Assume that $x_{i}$ evaluates to $u_{i}$ for $1\le i\le 3$, while the rest of the MIPE evaluates to $\alpha$. 
Then the first and fifth lines of the following table together give a contradiction:

\medskip

\begin{tabular}{cccccccccc}
  & $x_{1}$ & $x_{2}$ & $x_{3}$ & rest of MIPE & rest &  &  & & \\ \hline\hline
( & $u_{1}$ & $u_{2}$ &  $u_{3}$  & $\alpha$ & any &) & never $\in$ & $X$ & (since it is MIPE) \\\hline
( & $1-u_{1} $ & $u_{2}$ &   $u_{3}$ & $\alpha$ & some &) & $\in$ & $X$  & (since it was MIPE) \\
( & $u_{1}$ & $1-u_{2}$ &  $u_{3}$  & $\alpha$ & some &) & $\in$ & $X$  & (since it was MIPE) \\
( & $u_{1}$ & $u_{2}$ & $1 - u_{3}$  & $\alpha$ & some &) & $\in$ & $X$  & (since it was MIPE) \\\hline
( & $u_{1}$ & $u_{2}$ &   $u_{3}$     & $\alpha$ & some & ) & $\in$ & $X$  & ${\rm Maj}_{3}$(assgnms 2,3,4)  \\
\end{tabular}

\end{proof}
Assume now that we have an edge from $(k,\epsilon)$ to $(\ell,\epsilon')$ in the blockedness graph.
Since the MIPE creating this edge, by Lemma \ref{atmost2} has length two (cannot have length one),
we conclude that $x_k=\epsilon$ forces $x_{\ell} =\epsilon'$.
Consider any $x\in X$, by our assumption the total blockedness
graph is strongly connected so there is a path from $(1,x_{1})$ to 
$(1,1-x_{1})$, which in the light of the above argument means that $(1,x_{1})$
forces $(1,1-x_{1})$ through a sequence of edges, which is an obvious
contradiction. Thus $X$ is empty.

\medskip 

Schaefer's theorem now tells us that $X$ must be either an impossibility 
domain or an affine subspace.
So what about the case when $X$ is not totally blocked? Then we use the following deep theorem of A. Bulatov and P. Jeavons:

\begin{theorem}[Bulatov and Jeavons \cite{BulatovJ03}]
Let $D=\{0,1\}$ and $\Gamma$ be a set of multi-sorted relations on $D$ with type set $[t]$. 
Then for every type $j\in[t]$
either for every $\vec{f} = (f_{1},\cdots,f_{t}) \in {\rm MPol}{(\Gamma^{+})}$
the $j^{\rm th}$ component is a dictatorship or
there is an $\vec{f} = (f_{1},\cdots,f_{t}) \in {\rm MPol}{(\Gamma^{+})}$ such
that $f_{j}$ is one of 
\begin{enumerate}
\item the semi-lattice operation $u\vee v$ or $u\wedge v$, 
\item the majority operation $(u\vee v)\wedge (v\vee w)\wedge (w\vee u)$,
\item the Mal'tsev operation $u-v+w \mod 2$. 
\end{enumerate}
\end{theorem}
 
From this Theorem \ref{binaryclassification} easily follows.

\section{Non-binary Evaluations}

In this section we complete the classification theorem of Dokow and Holzman \cite{DH10b} for non-binary evaluations.
While the two authors have only considered the supportive case, here 
we resolve the idempotent case as well. 

\subsection{General notion of total blockedness}\label{totblsec}

Dokow and Holzman \cite{DH10b} have developed the the notion of total blockedness in the non-binary case, i.e. for a general domain $D$.
Our characterization theorems will use their notion.

\begin{definition}\label{totbl}
Total blockedness for non-binary $X$ is defined in \cite{DH10b}. $X$ is totally blocked if and only if the following directed graph on the vertex set 
\[
V = \{ \sigma\sigma'_{j} | \;\;\;  j \in [m] \; \; \wedge  \; \;
\sigma, \sigma' \in {\rm pr}_{j} X \;\; \wedge \;\; \sigma \neq \sigma' \}
\]
is strongly connected. There is an edge from $\sigma\sigma'_{k}$ to $\rho\rho'_{\ell}$ 
where $k \neq \ell$ if and only if there are
\[
B_{1} \subseteq  {\rm pr}_{1} X, \; B_{2} \subseteq  {\rm pr}_{2} X, \; \ldots\; , \; B_{m} \subseteq  {\rm pr}_{m} X
\]
such that 
\begin{enumerate}
\item Each $|B_j|=2$ ($1\le j\le m$);
\item $B_{k} = \{\sigma,\sigma'\}$ and $B_{\ell} = \{\rho,\rho'\}$;
\item Introduce the notation $X_{B} = X\cap \prod_{j=1}^{m}B_{j}$. Then $X_{B}$ is a binary relation
(think of each $B_{i}$ as $\{0,1\}$). Condition 3. is that under these definitions $(k,\sigma)$ is connected 
with $(\ell,\rho)$ in the blockedness graph of $X_{B}$ (see Definition \ref{blockgraph} for the blockedness graphs for binary domains).
\end{enumerate}
\end{definition}

\subsection{Supportive non-binary evaluations}\label{suppnonbin}

First we restate {\em the supportive} part of Theorem \ref{nonbinaryclassthm}:

\begin{theorem} \label{nonbinaryclassthm2}
Let $X\subseteq D^{m}$, non-degenerate and non-binary. 
If $X$ is totally blocked then $X$ is an impossibility 
domain with respect to IIA + Supportiveness + Non-dictatorship
if and only if $X$ is an impossibility 
domain with respect to IIA + Supportiveness + Non-dictatorship with the $n \leq 3$ condition
(meaning that we find no other aggregators than dictatorships when we restrict ourselves to aggregators with only three arguments).
\end{theorem}

\begin{remark}
The theorem is a classification theorem in the sense that it gives us
an algorithm to determine if a given domain $X$ is an impossibility domain with respect to IIA + Supportiveness + Non-dictatorship or not: just check all 
potential aggregators with at most three arguments. 
\end{remark}

\begin{proof} The `only if' part is straightforward since 
if $X$ is an impossibility domain with respect to IIA + Supportiveness + Non-dictatorship 
then every supportive IIA aggregator $f=(f_{1},\ldots,f_{m})$ of $X$ is a dictatorship, not only for arity $n\le 3$.
For the `if' part we need the following results of E. Dokow and R. Holzman: 

\begin{definition} [E. Dokow and R. Holzman \cite{DH10b}] \label{2dictatorshipdef}
Let $f=(f_{1},\ldots,f_{m})$ be a supportive IIA aggregator of arity $n$ for $X \subseteq D^m$. 
For an issue $j$ and an ordered pair 
of distinct positions $u,v \in D_{j}$ we translate $f_{j}|_{\{u,v\}}: \{u,v\}^{n}\rightarrow \{u,v\}$
to a function $W^{uv}_j:\{0,1\}^n\rightarrow\{0,1\}$ under $0\leftrightarrow v, \; 1\leftrightarrow u$.
\end{definition}

\begin{lemma} [E. Dokow and R. Holzman \cite{DH10b}, Propostion 1] \label{2neutrallemma}
If $X$ is totally blocked, then all $W_{j}^{uv}$s are the same. 
\end{lemma}

\begin{definition} [E. Dokow and R. Holzman \cite{DH10b}] \label{2dictatorshipdef}
We call $f$ {\em 2-dictatorial} if all $W^{uv}_j$s ($j\in [m]$, $u,v\in D_{j}$) are dictatorships with respect to the same $k$.
\end{definition}

\begin{lemma} [E. Dokow and R. Holzman \cite{DH10b}, Proposition 5] \label{2dictatorshiplemma}
If $X$ is totally blocked and $f$ is 2-dictatorial then $f$ is a dictatorship. 
\end{lemma}

Let us try to put the above two lemmas together! The total blockedness is a condition in the theorem 
whose `if' part we want to prove, so
what is missing is
that under the theorem's conditions
the $W^{uv}_j$s are not simply the same, but they are all dictatorships.
We rely on the following lemma which is essentially a 
simplified version of Theorem 2.16 in \cite{Bulatov11} and 
a multi-sorted version of Schaefer's dichotomy theorem (see Theorem \ref{schaeferalg}
and \cite{Chen09}).

\begin{lemma} \label{bulatovconservative}
Let $\Gamma$ be a set of multi-sorted relations over $D$ with 
type set $[t]$. If $\Gamma$ does not have IIA + Supportiveness + Non-dictator\-ship aggregators with the $n \leq 3$ condition, 
then for any 
$f=(f_1,\cdots,f_t) \in {\rm MPol(\Gamma^{\mho})}$,
$a \in [t]$, $u,v\in D_{a}$ with $u\neq v$ the restriction
$f_{a}|_{\{u,v\}}$ is a dictatorship.
\end{lemma}

\begin{proof} (of Lemma \ref{bulatovconservative}.)
Assume the contrary, namely that $\Gamma$ does not have IIA + Supportiveness + Non-dictator\-ship aggregators with $n \leq 3$, 
but it has some aggregator $f=(f_1,\cdots,f_t) \in {\rm MPol(\Gamma^{\mho})}$, such that 
there exist $a \in [t]$ and $u,v\in D_{a}$, $u\neq v$ with the property
that $f_{a}|_{\{u,v\}}$ is not a dictatorship, where $n$ is the arity of $f$. Then $n$ must be at least $4$.
Let $f$ be such a counter-example with minimal $n$. 
In particular, any $f' = f(x^{(1)},\ldots,\underbrace{x^{(i)}}_{i^{\rm th}\; {\rm argument}},\ldots,\underbrace{x^{(i)}}_{i'^{\rm th}\; {\rm argument}}\ldots,x^{(n)})$,
i.e. when we identify two inputs, must be a two-dictatorship, because $f'$ aggregates $n-1$ inputs.

\medskip

Denote $f_{a}|_{\{u,v\}}$ by $g$, which is presumably not a dictatorship. We will arrive at a 
contradiction by showing that $g$ is a dictatorship, i.e. there exists
one $k \in [n]$ such that for any $x_1,\cdots,x_n \in D'$, $g(x_1,\cdots,x_n)=x_k$. Since any identification of the variables of $g$
arises by first identifying these variables in $f$, and then restricting the resulting type = $a$ component to the binary set $\{u,v\}$, we
have, that any identification of variables of $g$, must result in a dictatorship. If \underline{all} of these identifications
$x^{(i)} = x^{(i')}$ result in a dictatorship that projects to coordinate $i$ (as opposed to some coordinate $i''\not\in
\{ i,i'\}$), we get a 
contradiction by setting $\{i,i'\}$ first to $\{1,2\}$ then to $\{3,4\}$:
\[
u = f(\underbrace{u,u},v,v,\ldots)= f(u,u,\underbrace{v,v},\ldots) =v. \;\;\;\; (\mbox{Used that $n\ge 4$}.)
\]
Thus there exist two coordinates where identifying the corresponding variables will result in a dictatorship 
function that projects to \underline{some other} (i.e. not $i,i'$) coordinate. Wlog assume that 
\[
g(\underbrace{x_1,x_1},x_3,x_4,\cdots,x_n)=x_4
\]
 Then
$g(x_1,x_1,x_1,x_4,\cdots,x_n)=x_4$. We show that this implies $g(x_1,x_2,x_1,x_4,\cdots,x_n)=x_4$.
If $g|_{x_{3}=x_{1}}$ was 
a dictator $x_{i}$ other than $x_{4}$,  then setting $x_2 = x_1$ and letting $x_{4}$ vary
we would get a contradiction. Similar reasoning gives that 
$g(x_1,x_3,x_3,x_4,\cdots,x_n)=x_4$. 
Thus whenever there is a duplication among the values of $x_{1},x_{2},x_{3}$, the output of $g$ is always $x_{4}$. But duplication 
always occurs since $|D'|=2$, thus $g$ is a dictatorship, a contradiction.
\end{proof}

We are now ready to prove the `if' part of Theorem \ref{nonbinaryclassthm2}. 
Assume that when $n\le 3$ there are no other supportive aggregators for $X$ than dictatorships (and the other conditions: 
$X$ is non-degenerate, totally blocked) also hold. Consider an aggregator $f$ for $X$ with $n\ge 4$.
By Lemma \ref{bulatovconservative}, 
for any $j \subseteq [m]$ and for any $u,v\in D_{j}$, $u\neq v$, 
we have $W_{j}^{uv}$ is a dictatorship. By Lemma \ref{2neutrallemma},  since $X$ is totally blocked, all 
$W_{j}^{uv}$  are the same. 
Lemma \ref{2dictatorshiplemma} then implies that $f$ is a dictatorship.
\end{proof}

\subsection{General idempotent non-binary evaluations} \label{secidempnonbinary}

 
\begin{lemma} \label{generalidempotent}
For a given domain $X \subseteq D^m$, non-degenerate and non-binary, 
i.e. $|D| = d \geq 3$, and totally blocked, if $X$ does not have any
IIA + Idempotency + Non-dictatorship 
aggregator with $n \leq d$ condition,  then $X$ is an impossibility 
domain with respect to IIA + Idempotency + Non-dictatorship.
\end{lemma}

\begin{proof}
Assume by contradiction that $X$ satisfies the condition of the lemma and $X$ is a 
possibility domain with respect to IIA + Idempotency + Non-dictatorship. 
Then there is an idempotent IIA aggregator $f=(f_1,\cdots,f_m)$ 
which is not a dictatorship. By the hypothesis of the lemma, 
$f$ has arity $n \geq d+1$. Assume that $f$ is of minimal arity, 
i.e. no idempotent non-dictatorial aggregator with smaller
arity exists. We first show that $f$ is supportive, i.e.
\[
\forall \;x^{(1)},\cdots,x^{(n)} \in X\;\;\;\;\;\;\; \forall \; 1 \leq j \leq m\; : \;\;\;\;\;\;\;\;\;\;\;\;\;\;\;
f_{j}(x^{(1)}_{j},\cdots,x^{(n)}_{j}) \in \{x^{(1)}_{j},\cdots,x^{(n)}_{j}\}. 
\]
We are done if we can show that for any \underline{fixed} $x^{(1)},\cdots,x^{(n)}$ and $j$
we have $f_{j}(x^{(1)}_{j},\cdots,x^{(n)}_{j}) \in \{x^{(1)}_{j},\cdots,x^{(n)}_{j}\}$.
We use the pigeon hole principle. Since $n \geq d+1$, among $x^{(1)}_{j},\cdots,x^{(n)}_{j}$ 
there must be at least two elements which are the same. 
For notational conveniences they are $x^{(1)}_{j}$ and $x^{(2)}_{j}$. 

\begin{remark} The collision does not happen at the same pair of indices for all inputs and $j$s,
but this does not affect us, since we are setting the input fixed.
\end{remark}

\medskip

Since the collision is at indices 1 and 2, we are going to examine the aggregator of $n-1$ elements $g$ 
that we get from $f$ by identifying the first two inputs of $f$. 
Since $g$ is also an idempotent IIA aggregator of $X$, 
by our minimality assumption $g$ must be a dictatorship. Thus $g_j$ is also a dictatorship and therefore supportive.
In particular, 
\[
u = g_{j}(x^{(2)}_{j},\cdots,x^{(n)}_{j}) \in \{x^{(2)}_{j},\cdots,x^{(n)}_{j}\}. 
\]
But $u$ is also the value that $f_{j}$ takes on $x^{(1)}_{j},x^{(2)}_{j}\cdots,x^{(n)}_{j}$ 
(since by our assumption $x^{(1)}_{j} = x^{(2)}_{j}$).
This concludes the proof of the fact that $f$ if supportive.

\medskip

Then, since we have found a non-dictatorial supportive aggregator for $X$ (on some number of inputs),
by Theorem \ref{nonbinaryclassthm2} there must also be a non-dictatorial supportive aggregator
on three inputs. Since supportive aggregators are also idempotent 
we get into a contradiction with the Lemma's assumption that 
the smallest non-dictatorial idempotent aggregator is on more than $|D| \ge 3$ inputs. 
\end{proof}

This lemma resolves the {\em idempotent case}, which was the remaining part of Theorem \ref{nonbinaryclassthm}.

\section{Algorithms to determine impossibility}

When we try to determine  if $X \subseteq D^m$
is an impossibility domain or not with respect to IIA + Idempotency   (or Supportiveness) + 
Non-dictatorship,
we can rely on two different types of characterization theorems:
By gadgets (Theorem \ref{gth}); by aggregators 
(Theorem \ref{nonbinaryclassthm}). 
Both types lead to algorithmic solutions, and we 
can use both of them as alternatives.  

\subsection{Algorithms from the characterization by gadgets}

As we have seen in the previous sections, to determine if $X$ is an impossibility domain with respect to 
IIA + Idempotency  (or Supportiveness) + Non-dictatorship, we need to check if 
certain relations can be expressed as  $X^{+}$-gadgets ($X^{\mho}$- gadgets).
This is stated in Theorem \ref{gth}.
The task is therefore to solve the following type of problem:

\medskip

{\bf Find Gadget:}  Given a set $\Gamma$ of 
relations over $D$ and a relation $R$ over $D$
determine if there is a $\Gamma$-gadget for $R$. 
In the multi-sorted version the relations are over a type set $[t]$ with associated domains $D_{1},\ldots,D_{t}$.

\bigskip

At first a hurdle seems to be that the gadgets may contain arbitrary (un-specified) number 
of auxiliary variables. Due to 
Geiger \cite{Geiger68}, Jeavons \cite{Jeavons98} 
and Trevisan et. al. \cite{TSSW96}, we however know that the
number of auxiliary variables in the {\em smallest} $\Gamma$-gadget for $R$
is upper bounded 
in terms of $|R|$, $|D|$ and $|\Gamma|$. The above
results easily generalize to the multi-sorted case. 
In our setting  $|\Gamma|$ is either $X^{+}$ or  $X^{\mho}$. The sizes of the latter are
upper bounded in terms of $m$ and $|D|$. Unfortunately, straightforward implementations yield a double-exponential running time
in $|X|$, $D$ and $m$, and whether we can
reduce it to single exponential, we currently do not know.

\medskip

In spite of their large worst case computational times, gadgets often provide very economical witnesses to impossibility
(see our two concrete earlier examples), and in all cases the witness size is at most
exponential in the above parameters.

\subsection{Algorithms from the characterization by aggregators}

Theorem \ref{nonbinaryclassthm} gives us a complete classification of impossibility domains 
for abitrary (binary or non-binary) evaluations. 
Here we turn these into algorithms. We assume without loss of generality that $X$ is non-degenerate.

\bigskip

{\bf Algorithm.} Determines if $X\subseteq D^{m}$ is an impossibility domain under the IIA + Supportiveness + Non-dictatorship condition:

\medskip

\begin{enumerate}
\item Check if $X$ is totally blocked (see Definition \ref{totbl}). 
If $X$ is not totally blocked, then $X$ is a possibility domain.

\medskip 

\hspace{-0.25in} {\em Otherwise:}

\item FOR each supportive $f=(f_{1},\ldots,f_{m})$ where $f_{j}: D_{j}^3 \to D_{j}$, check if $f$ aggregates $X$ and it is a Non-dictatorship.
If our search returns an $f$ satisfying the criteria, then $X$ is a possibility domain, otherwise it is an impossibility domain. 
\end{enumerate}

Note: it is enough to look only at $n=3$ as opposed to $n\le 3$.
The running time is dominated by step 2. The number of $f$s in the search is upper bounded by $|D|^{m \cdot |D|^3}$. 
To check for dictatorship is trivial. To check if $f$ is an aggregator of $X$ takes ${\rm poly(|X|,m, |D|)}$ time. Thus we have: 

\begin{theorem}
The complexity to determine if a given domain $X \subseteq D^m$ is an impossibility domain with respect to IIA + Supportiveness + Non-dictatorship is 
in $O( {\rm poly(|X|,m, |D|)}  |D|^{m \cdot |D|^{3}})$.
\end{theorem}



Similarly (we must replace 3 with $|D|$) we get:

\begin{theorem}
The complexity to determine if a given domain $X \subseteq D^m$ is an impossibility domain with respect to IIA + 
\underline{Idempotency} + Non-dictatorship is in 
$O( {\rm poly(|X|,m, |D|)} \cdot |D|^{m \cdot |D|^{|D|}})$.
\end{theorem}

Note that we view $D$ as constant (binary, etc.) and consequentially view the above complexities as single exponential.

\section{Degrees of democracy} \label{degreedemocracy}

Although the Non-dictatorship condition represents  minimal criterion for democracy, there are many functions that pass the 
Non-dictatorship test, but can barely be called democratic. Consider for instance the Boolean function that takes 
the majority value if the first voter votes zero, and takes the value one otherwise. Although this is not dictatorship the 
first voter has an overwhelming way in the outcome. What voting functions should we consider democratic?
Scenarios taken from real life, such as the American electoral system (iterated majority function),
show that the majority vote is not the only one that can be viewed as truly democratic. The answer is non-trivial.

\medskip

Different criteria for democracy have been formulated such as Anonymity (invariant under $S_{n}$)
or Symmetricity, by Kalai \cite{Kalai02} (invariant under 
a transitive permutation group acting on $[n]$), that are somewhere on the scale in
between the majoritarian and dictatorial voting schemes. 
In this article we would like to introduce StrongDem, with deep algebraic motivation.

\begin{description}
\item[StrongDem:] 
Let $f$ be an aggregator for $X\subseteq D^{m}$ for $n\ge 2$ voters that 
satisfies the IIA condition, so it
is of the form $f = (f_{1},\ldots,f_{m})$. We say that $f$ is StrongDem
if for every $1\le j\le m$, $1\le i\le n$ there are $c_{1},\ldots,c_{n}\in D_{j}$
such that $f_{j}(c_{1},\ldots,c_{i-1},x,c_{i+1},\ldots,c_{n})$ does not depend on $x\in D_{j}$.
We further require that this property holds not only for $D_{j}$ but when we replace $D_{j}$ with any $D_{j}'\subseteq D_{j}$ such that 
the operation $f$ preserves (respects) $D_{j}'$. (In the replacement the independence from $x$ must hold only when $x$ is also from $D_{j}'$.)
\end{description}

The majority function on three or more arguments 
is in StrongDem: take any $D'\subseteq D$ and set all votes except the vote of 
voter $i$ on the $j^{\rm th}$ issue to some (arbitrary) $c\in D'$. Then the outcome will be $c$
no matter what position the $i^{\rm th}$ voter takes. Since any StrongDem aggregator is
clearly Non-dictatorship, we have the containment:

\bigskip

\begin{center}
Non-dictatorship $\;\supset\;$ StrongDem $\;\supset\;$ Majority voting
\end{center}

\bigskip

All containments are strict in the following strong sense:
For any two of the above conditions we can find $X$ which is a possibility domain with respect to the 
larger class, but an impossibility domain for the smaller (IIA + Idempotency are assumed).
An important example for an $X$ which 
admits an $f$ with the Non-dictatorship condition,
but has no StrongDem voting scheme is the affine subspace.
Let $D=\{0,1\}$, and $X = \{(x_{1},x_{2},x_{3})\in D^{3}\mid\; x_{1}+x_{2}+x_{3} = 1 \; \mod 2\}$.
There is a non-dictatorial voting scheme when $n$ is odd: Let  
$f_{j}: (u_{1},\ldots,u_{n})\rightarrow \sum_{i=1}^{n} u_{i} \mod 2$ for $1\le j\le 3$.
It is easy to see that $f=(f_{1},f_{2},f_{3})$ has the Non-dictatorship property.  
It can be shown that this is the {\em only} type of non-dictatorial aggregator for $X$,
and it is not StrongDem.


\begin{figure}[htb] 
\begin{center}
\begin{tabular}{p{ 6.5cm}  p{7.5cm}}\hline
Majority, Anonymity, Symmetricity & The scheme treats all voters exactly in the same way \\\hline
StrongDem & When all others votes are appropriately fixed, a single voter cannot change the outcome \\\hline
Non-dictatorship &  There is not a single voter who exclusively controls the outcome. 
\end{tabular}
\caption{Conditions on democracy and their informal meaning}
\end{center}
\end{figure}


The StrongDem condition 
is equivalent to the aggregator falling into a well-researched class of universal algebraic operations.
This class contains operations with ``no ability to count''.
In contrast, functions like the parity function that have the ability to count are 
very input sensitive: their value changes even when the count changes only by one.
In a far-reaching part of the algebraic theory the ``no ability to count'' class of 
operations generate algebras that ``avoid types one and two'' \cite{HM88} congruences.
This class of algebras have been recently characterized by
local consistency checking algorithm, which was a breakthrough \cite{BK09}.
What makes the notion of StrongDem particularly attractive is that 
when viewing its minimalistic definition, it seems a {\em necessary} condition for democracy,
but it also has equivalent formulations, that are strong enough to accept it as a {\em sufficient} condition.

\begin{definition} [Strong resilience] 
Let $D$ be a finite domain and $\mu$ be a measure on $D$. The influence ${\rm Inf}_{i,\mu}(f)$ of the $i$-th variable of $f:D^n \to D$ is ${\rm Prob}_{\mu^{n+1}}(f(x) \neq f(x'))$, where $x,x'$ run through all random input pairs that differ only in the $i$-th coordinate ($\mu^{n+1}$ gives a natural measure on such pairs). The maximal influence $\max {\rm inf}_{\mu}(f)$ is $\max_{i} {\rm Inf}_{i,\mu}(f)$. A function $f:D^n \to D$ is strongly resilient if for every measure $\mu$ on $D$: $\max {\rm inf}_{\mu}(f^k) \to 0$ when $k \to \infty$ where $f^k$ is defined recursively by composition $f^k=f(f^{k-1},\ldots,f^{k-1})$.
\end{definition}

\begin{theorem} [G. Kun and M. Szegedy \cite{KunS09}]
The following are equivalent: 
\begin{enumerate}
\item $f$ is StrongDem.
\item There is a strongly resilient operation in $[\{f\}]$.
\end{enumerate}
\end{theorem}

\section{Final notes and future research}

Our translations yield further surprising findings. It turns out that Sen's famous theorem \cite{Sen79}
has a generalization readily taken off the shelve from \cite{BP75,JeavonsCC98}: 

\begin{theorem} \label{majorityth}
If $X\subseteq \{0,1\}^m$ is 2-decomposable and $k$ is odd, then the majority function with $k$
voters is an IIA aggregator for $X$. For $k$ = odd the converse also holds.
\end{theorem}

From this Sen's original theorem easily follows. In this work we have only started to elaborate on the close connection between 
Judgment Aggregation (more generally, Evaluation Aggregation) and universal algebra.
There are several other ramifications of the connection:
robustness results, extra conditions. In a follow-up result \cite{SX14}, we deal with the majoritarian aggregators. 

\section{Acknowledgements}
We would like to thank Andrei A. Bulatov and Ron Lavi for insightful comments.

\bibliographystyle{alpha}
\bibliography{arrows}

\end{document}